\documentclass[pdflatex,sn-apa]{sn-jnl}

\usepackage{graphicx}%
\usepackage{multirow}%
\usepackage{amsmath,amssymb,amsfonts}%
\usepackage{amsthm}%
\usepackage{mathrsfs}%
\usepackage[title]{appendix}%
\usepackage{xcolor}%
\usepackage{textcomp}%
\usepackage{manyfoot}%
\usepackage{booktabs}%
\usepackage{algorithm}%
\usepackage{algorithmicx}%
\usepackage[noend]{algpseudocode}%
\usepackage{listings}%

\usepackage{mathtools}%
\usepackage{bm}%
\usepackage{bbm}%
\usepackage{tabularx}%
\usepackage{makecell}%
\usepackage{subcaption}%
\usepackage{array}%
\usepackage[section]{placeins}%
\usepackage{float}%
\usepackage{enumitem}%

\theoremstyle{thmstyleone}%
\newtheorem{theorem}{Theorem}%
\newtheorem{proposition}[theorem]{Proposition}%
\newtheorem{lemma}[theorem]{Lemma}%
\newtheorem{corollary}[theorem]{Corollary}%

\newtheorem{remark}[theorem]{Remark}%

\theoremstyle{thmstylethree}%
\newcommand{\E}{\mathbb{E}}
\newcommand{\Prob}{\mathbb{P}}

\newcommand{\ind}{\mathbf{1}}
\newcommand{\dd}{\mathrm{d}}

\newcommand{\defeq}{\coloneqq}

\newcommand{\abs}[1]{\left\lvert #1 \right\rvert}

\newcommand{\Pmicro}{P^{\mathrm{micro}}}
\newcommand{\Plast}{P^{\mathrm{last}}}
\newcommand{\qmin}{q_{\min}}
\newcommand{\Dtot}{D_{\mathrm{tot}}}

\newcommand{\tdec}{t^{\mathrm{dec}}}
\newcommand{\Nsteps}{N_{\mathrm{steps}}}
\newcommand{\Nmax}{N_{\max}}
\newcommand{\dmax}{\delta_{\max}}
\newcommand{\Dtmin}{\Delta_{\min}^{\mathrm{dec}}}

\begin{document}

\title[The Impact of Market Informedness on Market Makers' Profitability]{The Impact of Market Informedness on Market Makers' Profitability}

\author*[1]{\fnm{Konrad} \sur{Ochedzan}}
\email{konrad.ochedzan@gmail.com}
\author[2,3]{\fnm{Nino} \sur{Antulov-Fantulin}}

\affil[1]{\orgdiv{Master of Science in Quantitative Finance}, \orgname{University of Zurich and ETH Zurich}, \orgaddress{\city{Zurich}, \country{Switzerland}}}
\affil[2]{\orgdiv{Computational Social Science, Department of Humanities, Social and Political Sciences}, \orgname{ETH Zurich}, \orgaddress{\city{Zurich}, \country{Switzerland}}}
\affil[3]{\orgname{Aisot Technologies}, \orgaddress{\city{Zurich}, \country{Switzerland}}}

\abstract{This paper studies how market informedness affects market makers’ profitability in a computational market environment with heterogeneous learning agents. We develop an agent-based market model in which market makers differ in their information sets and inventory-risk aversion, prices form endogenously, fundamental values evolve exogenously, and market-taker order flow follows a state-dependent self-exciting process.
The model provides a controlled computational laboratory for analyzing the interaction between informed trading, adverse selection, price discovery, and liquidity provision. We establish finite-horizon stability properties of the market-taker order-flow process and solve the market-making problem using multi-agent reinforcement learning with centralized training and decentralized execution.
The results show that informed market order flow is particularly harmful when aggregate market informedness is low, exposing market makers to severe adverse-selection risk. However, as market informedness increases, market-maker profitability displays an overall upward trend despite local non-monotonicities arising from complex market dynamics and stochastic learning. This suggests that the price-discovery benefits of informed trading can offset its adverse-selection costs.
The findings contribute to computational economics by showing how agent heterogeneity, endogenous price formation, and learning-based liquidity provision jointly shape market outcomes.}


\maketitle
\section{Introduction}

Modern financial markets are dynamic systems in which liquidity demand, liquidity supply, information, and strategic behavior interact at high frequency. At coarse time scales, asset prices are often represented by diffusion-type models, with the Black--Scholes framework providing a canonical benchmark for continuous-time price dynamics (\cite{BS1973}). At the transaction level, however, the relevant object is the market microstructure: trades and quotes are discrete, prices move in ticks, and the state of the market is summarized by the limit order book. This micro-level structure can generate volatility, clustering, and temporary deviations between transaction prices and latent economic value that are not visible in aggregate price models.

A central difficulty in modeling such markets is that order flow is both strategically and statistically heterogeneous. Behavioral feedback effects such as herding and belief-driven trading can generate bursts of activity and short-lived price dislocations (\cite{SH2000,GS2018}). Informational heterogeneity creates a separate but related mechanism: some traders possess, infer, or react more quickly to information about the asset's fundamental value, while others primarily condition on observable market data. This gives rise to adverse selection because liquidity providers are more likely to trade against informed market orders when quoted prices are temporarily misaligned with fundamentals (\cite{GLOSTENMILGROM1985,CAMPI2020,BMB2025}).

These mechanisms are especially important for market makers. Market makers supply liquidity by posting bid and ask limit orders, earn compensation through the bid--ask spread and exchange rebates, and face losses from inventory accumulation and adverse selection. Their profitability therefore depends not only on their own quoting policy, but also on the informational content of order flow, the endogenous response of prices, and the behavior of competing liquidity providers. This interaction is difficult to capture analytically because realistic market-making environments combine a discrete limit order book, clustered and size-heterogeneous market orders, partial information about fundamentals, and strategic competition among multiple agents.

This paper studies these interactions in a computational market environment. The model contains heterogeneous market makers with different information sets and inventory-risk aversion, a latent fundamental value, endogenous price formation through the limit order book, and market-taker order flow generated by a state-dependent marked Hawkes mechanism. The resulting environment is used as a controlled computational laboratory for analyzing how market informedness affects market-maker profitability through the joint channels of adverse selection and price discovery.

The learning component is based on reinforcement learning, which is well suited to sequential decision problems in large state spaces where closed-form optimal controls are unavailable (\cite{SZEPESVARI2010,SCHULMAN2017PPO,SUTTONBARTO2018}). Because the environment contains multiple strategic liquidity providers, the agents are trained using a multi-agent approach with centralized training and decentralized execution, a paradigm developed for interacting decision makers with non-stationary learning dynamics (\cite{LOWE2017MADDPG_CTDE,YU2022MAPPO}). This allows market makers to learn quoting policies from repeated interaction with the simulated market rather than from an exogenously specified rule.

The main finding is that market informedness has a non-trivial effect on liquidity provision. Informed market order flow is most damaging when aggregate market informedness is low, because market makers are then exposed to adverse-selection losses without sufficient price discovery. As aggregate informedness increases, market-maker profitability exhibits an overall upward trend despite local non-monotonicities caused by complex market dynamics and stochastic learning. The results suggest that the price-discovery benefits of informed trading can offset, and eventually dominate, its adverse-selection costs in a sufficiently informative market.

\section{Related Literature}

This paper relates first to the market-microstructure literature on spreads, dealer behavior, and liquidity provision. Classic models explain the bid--ask spread as compensation for transaction costs, immediacy services, and inventory risk faced by dealers and liquidity suppliers (\cite{DEMSETZ1968,TINIC1972,GARMAN1976,STOLL1978,AMIHUDMENDELSON1980,HOSTOLL1981}). Related work studies order-placement incentives and competition among liquidity suppliers, showing that the spread and execution outcomes depend on the strategic structure of the trading mechanism (\cite{CMSW1981,HOSTOL1983,MILDENSTEINSCHLEEF1983}). Information-based models add an adverse-selection component to the spread by showing that market makers protect themselves against the possibility of trading with better-informed counterparties (\cite{COPELANDGALAI1983,GLOSTENMILGROM1985}).

A second strand studies optimal market making and trading under inventory risk, signals, and partial information. The Avellaneda--Stoikov framework formulates high-frequency market making as an inventory-sensitive quoting problem around a reference price (\cite{AVELLANEDASTOIKOV2008}). Subsequent models incorporate order-flow signals, alpha signals, temporary and transient price impact, and speculative trading motives into optimal execution and market-making problems (\cite{CARTEAJAIMUNGAL2016,LEHALLENEUMAN2019,CARTEAWANG2020,DRISSI2022,NEUMANVOSS2022,BANKCARTEAKORBER2023}). Partial-information formulations explicitly connect market-making decisions to latent states and stochastic intensities, while recent work with fads, informed traders, and uninformed traders studies how the composition of market takers affects market-maker performance (\cite{CAMPI2020,BMB2025}). The present paper differs from this line by making transaction prices endogenous to the interaction of several heterogeneous market makers and a state-dependent stream of market orders.

The paper also builds on statistical models of order-flow dynamics. Hawkes processes provide a tractable representation of self-exciting arrivals, where one event temporarily increases the likelihood of subsequent events (\cite{HA1971}). Intensity-based point-process models have been used to describe security-market events in continuous time (\cite{BOWSHER2007}). Empirical studies of limit order books document clustered activity, dependence in order signs, and rich volume dynamics (\cite{CHAKRABORTITOKEETAL2011,GOULDEtAl2013}). Hawkes models have been applied to high-frequency market data and systemic price co-jumps, capturing temporal clustering beyond Poisson order-arrival assumptions (\cite{BORMETTIETAL2015,FILIMONOVSORNette2015}). Marked and multivariate extensions incorporate order size and cross-excitation across event types, which is important because large market orders can walk the book and generate discontinuous quote changes (\cite{RMB2017,JAIN2024}). Our model uses this tradition to generate clustered and size-heterogeneous market-taking events, but modulates their intensities by the relationship between the market price and a latent fundamental value.

Finally, the paper contributes to computational and agent-based approaches to financial markets. Agent-based market simulations have been used to study how heterogeneous traders and trading rules can generate aggregate market phenomena that are difficult to derive in closed form (\cite{PADDRIKETAL2012CIFER}). Algorithmic-trading references provide the broader market-design and execution context for these simulation-based approaches (\cite{CARTEAJAIMUNGALPENALVA2015_AHFT}). Reinforcement learning provides a general framework for sequential decision-making and policy optimization in complex environments (\cite{SZEPESVARI2010,SCHULMAN2017PPO,SUTTONBARTO2018}). Multi-agent reinforcement learning extends this framework to interacting agents whose learning problems are coupled through the environment and through each other's policies (\cite{LOWE2017MADDPG_CTDE,YU2022MAPPO}). Recent applications to market making and limit order books show that learning-based agents can be used to analyze quoting behavior, equilibrium formation, and endogenous price dynamics in simulated markets (\cite{Gasperov2021RLMM,HeLin2022RLEquilibriumLOB}). Closely related computational market environments combine deep learning, queue-reactive order-book simulation, and multi-agent reinforcement learning to study realistic limit-order-book behavior (\cite{BODORCARLIER2025,cheridito2025abidesmarl}). Relative to these studies, the present paper focuses on how market informedness changes the balance between adverse selection and price discovery for heterogeneous learning market makers.

\section{Contributions}
This paper makes three contributions.

First, it develops a computational market environment with heterogeneous market-making agents, endogenous price formation, state-dependent self-exciting order flow, and learning-based decision-making.

Second, it provides stability guarantees for the resulting market-taker process, ensuring that the simulated market dynamics are well defined over finite horizons.

Third, it uses this environment to study how market informedness affects liquidity provision. The results show that informed order flow is most harmful in poorly informed markets, while higher aggregate informedness can improve market-maker profitability through better price discovery.

\section{Computational Market Environment}

This section describes the simulation environment used throughout the paper. The objective is to obtain an order-driven market in which (i) liquidity demand arrives in continuous time with empirically realistic clustering, (ii) trades have heterogeneous sizes, and (iii) arrival rates react to both microstructure conditions (spread, depth) and mispricing relative to an exogenous fundamental value. 

\subsection{Notations}

Market makers act on a fixed decision grid  and interact through a shared limit order book (LOB), while market takers generate trades in event time. Let $0=t^{\mathrm{dec}}_0<t^{\mathrm{dec}}_1<\dots<t^{\mathrm{dec}}_{N_{\text{steps}}}=T$ denote agents' decision times. We assume the following notations.

Orders end up in a \emph{queue} which accounts for the order of limit orders. That leads to a situation, where at each point of time there is a finite list of available prices at each side. We denote a list of open asks' price levels at time $t$ as $\mathcal{A}_t$ and bids' as $\mathcal{B}_t$. In the LOB for each price $p$ the book maintains a first-in-first-out (FIFO) queue (price--time priority) of resting orders:
\begin{align}
\mathcal{Q}^b_t(p)=\big((j^b_{1,\ p},q^b_{1,\ p}),(j^b_{2,\ p},q^b_{2,\ p}),\dots, (j^b_{n_p,\ p},q^b_{n_p,\ p})\big)_t,\\
\mathcal{Q}^a_t(p)=\big((j^a_{1,\ p},q^a_{1,\ p}),(j^a_{2,\ p},q^a_{2,\ p}),\dots,(j^a_{n_p,\ p},q^a_{n_p,\ p})\big)_t,
\label{def:queue}
\end{align}

where $(j^x_{m,\ p},q^x_{m,\ p})$ denotes an order of size $q^x_{m,\ p}$ posted by entity $j^x_{m,\ p}$ on side $x\in\{\text{bid}, \text{ask}\}$ and price level $p$. Importantly, all prices live on a fixed tick grid with tick size $\Delta p$. For the simulation purposes to make sense out of agents actions and prevent floating point drift, every price that enters the book is first \emph{canonicalized} by rounding to the appropriate number of decimals implied by $\Delta p$. This guarantees that identical economic price levels map to the same dictionary key. Throughout, we denote by $\mathrm{canon}(\cdot)$ this rounding operation. Thanks to that, in cases of odd length spread, the offers stay on the tick grid. 

Whenever a market order fills a limit order, the market maker who quoted that offer may (in the implementation does) receive a rebate $R_{\text{rebate}} >0$. That serves as additional motivation to take the risks of market making.

The \emph{best bid} is the highest available buy price; the \emph{best ask} is the lowest available sell price. Their difference is the \emph{bid--ask spread}, and their average is the \emph{mid-price}. The total available quantity at each level constitutes the \emph{depth}. In this paper the best bid is denoted as $B_t=\max\{p:\,p\in\mathcal{B}_t\}$ and the best ask as $A_t=\min\{p:\,p\in\mathcal{A}_t\}$. The liquidity at best ask and liquidity at best bid are denoted respectively as $Q^a_t = \sum_{(j_i,q_i) \in {Q}^a_t(A_t)}q_i $ and $Q^b_t = \sum_{(j_i,q_i) \in {Q}^b_t(B_t)}q_i$. 

We model the market price defined as last transaction price $P^{\mathrm{last}}_t$ as the last price at which a trade occurred before or at time $t$.  Market mid is defined as:
\begin{align}
    M_t = \frac{B_t + A_t}{2},
\end{align}
and market micro-price as:
\begin{align}
P^{\mathrm{micro}}_t=\frac{A_t\,Q^b_t + B_t\,Q^a_t}{Q^a_t+Q^b_t},
\label{eq:microprice}
\end{align}
whenever both sides of the book are non-empty. If either side is empty, we fallback to the last transaction price $P^{\mathrm{last}}_t$ which is the last price at which a trade occurred before or at time $t$.

Each market maker aggregates both inventory and cash which for agent $i$ at time $t$ we denote as respectively  $q_t^{(i)}$ and $c_t^{(i)}$.

\subsection{Fundamental value}\label{sec:fundamental_value}
The fundamental value is modeled as an exogenous process. Historically, asset prices have often been modeled by geometric Brownian motion or by richer specifications designed to reproduce stylized volatility patterns. The present paper, however, focuses on intraday market microstructure rather than on long-horizon asset pricing. For this reason, and similarly to the role played by the reference-price process in \cite{AVELLANEDASTOIKOV2008}, we model the fundamental value by a standard Brownian motion. This choice keeps the specification parsimonious and is appropriate for the short horizons considered here, over which both drift and the risk-free rate can be treated as negligible. In addition, introducing a process with an average one-sided trend would create stronger incentives for agents to take directional bets, thereby generating inventory and wealth distributions that are less representative of intraday market making. Therefore, we model the fundamental as a Brownian motion with zero drift:
\[
S_t = S_0 + \sigma W_t,
\]
where $W_t$ is a Brownian motion, $S_0$ is the initial price, and $\sigma >0$ is the volatility. Since an arithmetic Brownian motion can in principle cross zero, the simulator imposes an economically natural lower bound at zero: in the implementation, zero is treated as a hard floor and the fundamental price is not allowed to fall below it.

We also introduce a notion of a ``noisy estimator'' of the fundamental value so that some agents can be given additional information about the market. This noisy estimator is denoted by $\widetilde{S}_t^{(i)}$ and defined as $\widetilde{S}_t^{(i)} = S_t + \varepsilon_t^{(i)}$, where $\varepsilon_t^{(i)} \sim \mathcal{N}(0,\sigma_{signal}^{(i)})$ is an exogenous, agent-specific noise term.

\subsection{Market takers}

We model market taking events as a multivariate marked point process. Following \cite{RMB2017}, the mark jointly specifies the trade direction (consuming the ask side or the bid side) and a discrete size bin. Let $K$ denote the number of size bins on each side and define $C=2K$. We index event types by
\[
i \in \{1,\dots,K\}\quad\text{(market buy events consuming asks)}, 
\]
\[
i \in \{K+1,\dots,2K\}\quad\text{(market sell events consuming bids)}.
\]
For each type $i$ we introduce a counting process $N_i(t)$ and an $\mathcal{F}_t$-predictable intensity $\lambda_i(t)$, where $\mathcal{F}_t$ is the natural filtration of the process.

\paragraph{Raw Hawkes dynamics.}
To capture self-excitation, we first define a \emph{raw} Hawkes intensity vector $\boldsymbol{\lambda}^{H}(t)=(\lambda^{H}_1(t),\dots,\lambda^{H}_{C}(t))^\top$ as in \cite{HA1971}:
\begin{equation}
\lambda^{H}_i(t)
=
\mu_i
+
\sum_{j=1}^{C}\int_{0}^{t}\phi_{i\leftarrow j}(t-s)\,\mathrm{d}N_j(s),
\qquad i=1,\dots,C,
\label{eq:hawkes_raw_vector}
\end{equation}
with baseline rates $\mu_i>0$ and excitation kernels $\phi_{i\leftarrow j}(\cdot)\ge 0$. As standard in microstructure Hawkes modeling and consistent with \cite{RMB2017}, we assume a sum-of-exponentials parameterization:
\begin{equation}
\phi_{i\leftarrow j}(u)=\sum_{r=1}^{R}\alpha_{ijr}\,e^{-\beta_{ijr}u},
\qquad u\ge 0,
\label{eq:exp_kernel}
\end{equation}
where $\alpha_{ijr}\ge 0$ and $\beta_{ijr}>0$. This specification yields a piecewise-deterministic representation via auxiliary excitation states
\[
Z_{ijr}(t)=\int_{0}^{t}\alpha_{ijr}e^{-\beta_{ijr}(t-s)}\,\mathrm{d}N_j(s),
\qquad
\lambda^{H}_i(t)=\mu_i+\sum_{j=1}^{C}\sum_{r=1}^{R}Z_{ijr}(t),
\]
with exponential decay between jumps and instantaneous increments at event times.

\paragraph{Microstructure and information modulation.}
The raw Hawkes dynamics account for clustering and cross-excitation across sizes and sides, but do not incorporate the dependence of trade arrivals on current market conditions or on mispricing relative to fundamentals. We introduce this dependence through a multiplicative gain vector $G(x(t))\in\mathbb{R}_+^{C}$, where $x(t)$ is a collection of the LOB and fundamental-value features. We then construct:

\begin{align}
\delta^a_t
&\defeq A_t-P^{\mathrm{micro}}_t \ge 0,
\qquad
\delta^b_t
\defeq
P^{\mathrm{micro}}_t-B_t \ge 0,\qquad
I_t
\defeq S_t-P^{\mathrm{micro}}_t.
\label{eq:deltas}
\end{align}

The gain $G$ reflects two effects: (i) sensitivity to instantaneous spread conditions captured by $\delta^a_t,\delta^b_t$ and (ii) information-driven trading pressure when the observed market price deviates from the fundamental value via $I_t$. Following the mixture-of-informed approach in \cite{BMB2025}, we assume that a fraction $\psi\in[0,1]$ controls the market informedness. We define spread sensitivity $\eta > 0$ and mispricing sensitivity $\zeta >0$ equal for both sides and all size bins. We then set the uninformed and informed gains as
\begin{align}
g^{U,a}(x(t)) &= \exp\!\big(-\eta\,\delta^a_t\big), &
g^{I,a}(x(t)) &= \exp\!\big(-\eta\,\delta^a_t + \zeta\,I_t\big), \label{eq:gains_ask}\\
g^{U,b}(x(t)) &= \exp\!\big(-\eta\,\delta^b_t\big), &
g^{I,b}(x(t)) &= \exp\!\big(-\eta\,\delta^b_t - \zeta\,I_t\big). \label{eq:gains_bid}
\end{align}
The effective gains (mixture over trader types) are
\begin{equation}
G^a(x(t))=\psi\,g^{I,a}(x(t))+(1-\psi)\,g^{U,a}(x(t)),
\qquad
G^b(x(t))=\psi\,g^{I,b}(x(t))+(1-\psi)\,g^{U,b}(x(t)),
\label{eq:gains_mixture}
\end{equation}
and we construct $G(x(t))\in\mathbb{R}^C$ by concatenation:
\[
G_i(x(t))=
\begin{cases}
G^a(x(t)), & i=1,\dots,K,\\
G^b(x(t)), & i=K+1,\dots,2K.
\end{cases}
\]

If one side of the book is empty, then the corresponding gains are set to zero.

The \emph{effective} intensity used for simulation is:
\begin{equation}
\lambda_i(t) = G_i(x(t))\,\lambda^{H}_i(t),
\qquad i=1,\dots,C,
\label{eq:effective_intensity}
\end{equation}
and the total event rate is
\begin{equation}
\Lambda_{\mathrm{eff}}(t)=\sum_{i=1}^{C}\lambda_i(t)=\sum_{i=1}^{C}G_i(x(t))\,\lambda_i^{H}(t).
\label{eq:lambda_total}
\end{equation}

That notion is a natural extension of \cite{BMB2025} where we would take $C=2$ for only unitary order volumes at both sides, null Hawkes kernels and constant baseline Hakwes intensities $\mu_i$. Therefore, this research serves as an extension in a more general setup to the aforementioned paper.

\paragraph{Size realization within bins.}
An event type determines only the side and the size bin. Conditional on an event in bin $k$, the executed quantity is sampled as an integer uniformly from a pre-specified range associated with that bin. This binning is the mechanism that produces heterogeneous order sizes while keeping the Hawkes dimension finite and relatively condensed, as in \cite{RMB2017}.

\paragraph{Warm start}
To avoid an artificial burn-in period at the beginning of each episode, we initialize the raw Hawkes state in a stationary regime of the linear Hawkes model \eqref{eq:hawkes_raw_vector}--\eqref{eq:exp_kernel}. Let $A\in\mathbb{R}^{C\times C}$ be the integrated kernel matrix with entries
\begin{equation}
A_{ij}=\int_{0}^{\infty}\phi_{i\leftarrow j}(u)\,\mathrm{d}u
=\sum_{r=1}^{R}\frac{\alpha_{ijr}}{\beta_{ijr}}.
\label{eq:integrated_kernel}
\end{equation}

If the spectral radius $\rho(A)<1$, the stationary mean intensity $\boldsymbol{\lambda}^*=\mathbb{E}[\boldsymbol{\lambda}^H(t)]$ exists and satisfies
\[
\boldsymbol{\lambda}^{*}
=
\boldsymbol{\mu}+A\,\boldsymbol{\lambda}^{*},
\qquad\text{hence}\qquad
\boldsymbol{\lambda}^{*}=(I-A)^{-1}\boldsymbol{\mu}.
\]
In the simulation, we set $\boldsymbol{\lambda}^{H}(0):=\boldsymbol{\lambda}^*$ and initialize the excitation variables $Z_{ijr}(0)$ so that $\lambda_i^{H}(0)=\mu_i+\sum_{j,r}Z_{ijr}(0)$. Concretely, we place the ``excess over baseline'' $\lambda_i^*-\mu_i$ on the diagonal $j=i$ and split it evenly across the $R$ exponential components:
\[
Z_{ijr}(0)=
\begin{cases}
\dfrac{\lambda_i^{*}-\mu_i}{R}, & i=j,\\[4pt]
0, & i\neq j.
\end{cases}
\]
This choice ensures $\lambda_i^{H}(0)=\lambda_i^*$ while keeping the initialization simple and reproducible.

\paragraph{Continuous time simulation}

Market-taker arrivals are simulated in continuous time using an event-driven procedure based on Ogata's thinning / time-change methodology (\cite{OGATA1981,DALEYVEREJONES2003}). The market makers, in contrast, act on a fixed decision grid. Over each interval $(t^{\mathrm{dec}}_n,t^{\mathrm{dec}}_{n+1})$, where $t^{\mathrm{dec}}_n$ is agents' n-th decision time, agents do not submit new quotes, so the control is held fixed. However, taker trades occurring within the interval consume resting depth and may change the best quotes and queue sizes. Accordingly, the effective intensities are piecewise defined between event times: between consecutive events the raw Hawkes state decays deterministically, while after each executed trade the LOB-derived multiplier $G(x(t))$ is recomputed from the updated book.

after each event we may draw $S\sim\mathrm{Exp}(1)$ and define the next event time as the solution of
\begin{equation}
\int_{t}^{t_{\mathrm{next}}}\Lambda_{\mathrm{eff}}(s)\,\mathrm{d}s = S.
\label{eq:time_change_equation}
\end{equation}

Between events, the Hawkes excitation variables decay exponentially:
\[
Z_{ijr}(t+u)=Z_{ijr}(t)\,e^{-\beta_{ijr}u},\qquad u\ge 0,
\]
so the raw intensities evolve as
\[
\lambda_i^{H}(t+u)=\mu_i+\sum_{j=1}^{C}\sum_{r=1}^{R}Z_{ijr}(t)\,e^{-\beta_{ijr}u}.
\]
Within a time segment over which $G(x(\cdot))$ is held fixed, the integrated effective intensity admits a closed form. Denoting the segment length by $\Delta$, we obtain
\begin{align}
\int_{t}^{t+\Delta}\Lambda_{\mathrm{eff}}(s)\,\mathrm{d}s
&=
\sum_{i=1}^{C}G_i(x(t))
\left[
\mu_i\,\Delta
+
\sum_{j=1}^{C}\sum_{r=1}^{R}
Z_{ijr}(t)\,
\frac{1-e^{-\beta_{ijr}\Delta}}{\beta_{ijr}}
\right].
\label{eq:integrated_intensity_closed_form}
\end{align}
Therefore, on each segment the mapping $\Delta\mapsto \int_{t}^{t+\Delta}\Lambda_{\mathrm{eff}}(s)\,\mathrm{d}s$ is continuous and strictly increasing, which makes \eqref{eq:time_change_equation} well-posed. In the implementation, the root-finding step is performed by bisection on $\Delta$ using \eqref{eq:integrated_intensity_closed_form}.

We do not redraw the exponential threshold at every market-maker decision time. Instead, we maintain a remaining ``area to the next event'' $S^{\mathrm{rem}}$ corresponding to the left-hand side of \eqref{eq:time_change_equation}. On each time segment $[t,t_{\mathrm{end}}]$ over which the multiplier $G(x(\cdot))$ is held fixed and where $t_{\text{end}}$ stands for the next agents' decision step, we define the integrated hazard by
\[
H(t,t_{\mathrm{end}})
:=
\int_t^{t_{\mathrm{end}}}\Lambda_{\mathrm{eff}}(s)\,\mathrm{d}s.
\]

Thus, $H(t,t_{\mathrm{end}})$ is precisely the cumulative compensator accumulated over the current segment.
\begin{itemize}
\item If $H(t,t_{\mathrm{end}}) < S^{\mathrm{rem}}$, then no trade occurs before $t_{\mathrm{end}}$. We advance the clock to $t_{\mathrm{end}}$, update $G$ according to market-maker actions at $t_{\mathrm{end}}$, set $S^{\mathrm{rem}}\leftarrow S^{\mathrm{rem}}-H(t,t_{\mathrm{end}})$, and proceed.
\item If $H(t,t_{\mathrm{end}}) \ge S^{\mathrm{rem}}$, then a trade occurs inside $(t,t_{\mathrm{end}}]$. We solve for $t_{\mathrm{next}}$ such that \newline $\int_t^{t_{\mathrm{next}}}\Lambda_{\mathrm{eff}}(s)\,\mathrm{d}s=S^{\mathrm{rem}}$, advance to $t_{\mathrm{next}}$, and sample the event type. The Hawkes state is then updated by applying a jump in the realized dimension:
\[
Z_{ijr}(t_{\mathrm{next}}^+) = Z_{ijr}(t_{\mathrm{next}}) + \alpha_{ijr}\,\mathbf{1}\{j=\text{realized type}\}.
\]
After executing the corresponding market order in the LOB, we draw a fresh $S^{\mathrm{rem}}\sim\mathrm{Exp}(1)$ representing the waiting time for the next taker event and continue.
\end{itemize}
If the next market-maker decision time lies before the unconstrained solution to \eqref{eq:time_change_equation}, the procedure above naturally splits the integrated intensity across segments and carries $S^{\mathrm{rem}}$ forward. This coupling ensures that market-taker trades occur in true continuous time, while market-maker actions remain on the fixed decision grid, and the exponential time-change mechanism stays consistent across all regime and state updates.

\section{Stability of the modulated Hawkes mechanism}\label{sec:hawkes_stability}

The multiplicative gains in \eqref{eq:gains_mixture} make the market-taker process state dependent and potentially much more active than a standard linear Hawkes process. In particular, when the microprice is below the fundamental value, buy-side gains are amplified, and when the microprice is above the fundamental value, sell-side gains are amplified. This raises a stability question: can the exponential gain in $|I_t|$ lead to explosive order arrivals or uncontrolled mispricing within one episode? The following results show that, on finite horizons, the finite depth of the LOB provides a stabilizing feedback mechanism. The amplified side consumes finite resting depth and therefore cannot generate infinitely many trades before the next decision time. Detailed proofs and auxiliary lemmas are given in Appendix~\ref{app:stability_proofs}.

For the stability statements, call an event type \emph{corrective} at time $t$ if it trades in the direction that is favored by the sign of the pre-jump mispricing: buy types $i\le K$ are corrective when $I_{t^-}\ge0$, and sell types $i>K$ are corrective when $I_{t^-}<0$. The other types are called \emph{destabilizing}. Let $T_{\mathrm{ecs}}$ denote the total time during $[0,T]$ at which the corrective side of the book is empty:
\begin{equation}
  T_{\mathrm{ecs}}
  :=
  \int_0^T
  \ind_{\{\text{corrective side empty at }t^-\}}\,\dd t .
  \label{eq:tecs}
\end{equation}

\paragraph{Assumptions for the stability analysis.}\label{par:stability_conditions}
For the stability results in this subsection, we work on a fixed finite horizon $[0,T]$ and use the following assumptions. These assumptions are in accordance with the market environment described above: they formalize the depth, execution, quote-range, and fundamental-value specifications imposed by the simulator, rather than additional restrictions on the learned quoting policies. Assumptions \textup{(S1)--(S6)} are sufficient for non-explosion and the exponential mispricing bounds; the pathwise tail bound additionally uses the Brownian fundamental specification in \textup{(S7)}.
\begin{description}[leftmargin=2.7em,style=nextline]
\item[(S1) Positive gain parameters.]
The gain parameters satisfy $\eta>0$, $\zeta>0$, and $\psi\in(0,1]$.
\item[(S2) Finite posted depth.]
At every decision time $\tdec_n$, the total resting quantity on each side of the book is bounded by a deterministic constant $\Dtot<\infty$. This bound is on aggregate depth across all market makers.
\item[(S3) Minimum executable size.]
There is a constant $\qmin>0$ such that every executed market order removes at least $\qmin$ units from the corresponding side. If less than $\qmin$ units are available on a side, then no market order against that side is executed.
\item[(S4) Empty-side shutdown.]
An empty side cannot be hit by further market orders: if the ask side is empty, then $G^a=0$, and if the bid side is empty, then $G^b=0$.
\item[(S5) Finite decision grid.]
The number of market-maker decision times is finite and denoted by $\Nsteps$. When a rate per unit time is required, the grid is assumed to satisfy $\tdec_{n+1}-\tdec_n\ge\Dtmin>0$.
\item[(S6) Bounded quote distances.]
There exists $\dmax<\infty$ such that, whenever the relevant side of the book is non-empty,
\[
  0\le\delta^a_t\le\dmax,
  \qquad
  0\le\delta^b_t\le\dmax .
\]
In the implementation, this bound is maintained by restricting the admissible posting range and cancelling quotes that move outside it after continuous-time taker trading.
\item[(S7) Brownian fundamental value.]
The fundamental value follows the Brownian specification introduced in Section~\ref{sec:fundamental_value},
\[
  S_t=S_0+\sigma W_t,
  \qquad \sigma>0,
\]
where $W$ is a standard Brownian motion.
\end{description}

\begin{proposition}[Non-explosion]\label{prop:nonexplosion}
Set
\begin{equation}
  \Nmax \defeq \frac{2\Dtot\Nsteps}{\qmin}.
  \label{eq:Nmax}
\end{equation}
Under Assumptions \textup{(S1)--(S6)}, the total number of market-taker events satisfies
\begin{equation}
  N([0,T]) \defeq \sum_{i=1}^{C}N_i([0,T]) \le \Nmax
  \qquad\text{a.s.}
  \label{eq:nonexplosion_bound}
\end{equation}
In particular, $\E[N([0,T])]\le\Nmax$.
\end{proposition}
Economically, finite posted depth, minimum trade sizes, and empty-side shutdowns prevent an infinite cascade of market orders over a finite trading horizon. This is usefull in mathematical modelling of the market as well as grants a decent bound on what maximal level of activity can be expected under given market conditions.
Proposition~\ref{prop:nonexplosion} is an extension of standard non-explosion results for Hawkes processes to the environment presented in this paper. Let $\mu_{\min}:=\min_{1\le i\le C}\mu_i$ and
\begin{equation}
  \gamma := \psi\,e^{-\eta\dmax}\,\mu_{\min}.
  \label{eq:gamma}
\end{equation}
This is the uniform coefficient in the lower bound for corrective intensity during periods in which the corrective side of the book is non-empty.

\begin{proposition}[Exponential mispricing integrability]\label{prop:exp_int}
Under Assumptions \textup{(S1)--(S6)},
\begin{equation}
  \E\!\left[\int_0^T e^{\zeta\abs{I_{t^-}}}
    \ind_{\{\text{corrective side non-empty at }t^-\}}\,\dd t\right]
  \le \frac{\Nmax}{\gamma}.
  \label{eq:exp_int}
\end{equation}
Moreover, on the full interval,
\begin{equation}
  \E\!\left[\int_0^T e^{\zeta\abs{I_{t^-}}}\,\dd t\right]
  \le
  \frac{\Nmax}{\gamma}
  +
  \E\!\left[e^{\zeta\sup_{t\le T}\abs{I_{t^-}}}T_{\mathrm{ecs}}\right].
  \label{eq:exp_int_tecs}
\end{equation}
The second term is the contribution of periods during which the corrective side is empty. 
\end{proposition}
Economically, the result captures the model's price-discovery force: when the microprice deviates from fundamentals and corrective liquidity is available, informed market orders discipline the deviation, making adverse selection for market makers part of the same mechanism that stabilizes prices.

\begin{corollary}[Occupation-time bound]\label{cor:occupation}
For every $M>0$,
\begin{equation}
  \E\!\left[\int_0^T
  \ind_{\{\abs{I_{t^-}}>M\}}
  \ind_{\{\text{corrective side non-empty at }t^-\}}\,\dd t\right]
  \le \frac{\Nmax}{\gamma}e^{-\zeta M}.
  \label{eq:occupation}
\end{equation}
Thus, outside corrective-empty periods, the expected occupation time of mispricing above level $M$ decays exponentially in $M$. On the full interval,
\begin{equation}
  \E\!\left[\int_0^T
  \ind_{\{\abs{I_{t^-}}>M\}}\,\dd t\right]
  \le
  \frac{\Nmax}{\gamma}e^{-\zeta M}+\E[T_{\mathrm{ecs}}],
  \label{eq:occupation_full}
\end{equation}
where heuristically under competition and optimality $\E[T_{\mathrm{ecs}}] \approx 0$.
\end{corollary}
This corollary means that large mispricing are expected to be short, with durations that decrease exponentially in the threshold except during temporary liquidity dry-ups when the corrective side of the book is absent then any bound requires some prior knowledge about the posting strategy of agents. Nevertheless, in any sensible market the time spent with one side of the book empty is negligible leading to an approximate bound on the total time mispricing spends being large.

The next statement controls the entire path, including corrective-empty periods.

\begin{proposition}[Pathwise mispricing tail bound]\label{prop:pathwise}
Under Assumptions \textup{(S1)--(S7)}, let
\begin{equation}
  \Delta_P := 2\dmax(\Nmax+\Nsteps),
  \qquad
  K_0 := \abs{I_0}+\Delta_P
  = \abs{S_0-\Pmicro_0}+2\dmax(\Nmax+\Nsteps).
  \label{eq:K0}
\end{equation}
Then, for every $K>K_0$,
\begin{equation}
  \Prob\!\left(\sup_{t\in[0,T]}\abs{I_t}\le K\right)
  \ge
  1-2\exp\!\left(-\frac{(K-K_0)^2}{2\sigma^2T}\right).
  \label{eq:pathwise_tail}
\end{equation}
In particular, $\Prob(\sup_{t\le T}\abs{I_t}<\infty)=1$.
\end{proposition}
This theorem means that the endogenous price process cannot generate unbounded dislocations over the simulation horizon: the worst deviation is controlled by finite LOB movements and by the volatility of the exogenous fundamental value. It allows to access the tail behaviour of the mispricing.

\begin{remark}[Moments of the supremum]\label{rem:moments}
As a standard result, the Gaussian tail bound~\eqref{eq:pathwise_tail} implies finiteness of all polynomial moments of $\sup_{t\le T}\abs{I_t}$. In particular, for every $p\ge1$,
\begin{align*}
  \E\!\left[\sup_{t\le T}\abs{I_t}^p\right]
  &\le
  K_0^p
  +p\int_{K_0}^{\infty}K^{p-1}
    2\exp\!\left(-\frac{(K-K_0)^2}{2\sigma^2T}\right)\dd K\\
  &<\infty.
\end{align*}
Moreover, for every $\lambda<1/(2\sigma^2T)$,
\begin{equation*}
  \E\!\left[
  \exp\!\left(\lambda\bigl(\sup_{t\le T}\abs{I_t}-K_0\bigr)^2\right)
  \right]<\infty .
\end{equation*}
\end{remark}

These bounds formalize the stabilizing mechanism of the extended Hawkes model. Large mispricing increases the intensity of corrective flow exponentially, but corrective flow consumes finite depth. Consequently, periods with large active corrective pressure are short in compensator time, and even periods without an active corrective side remain pathwise controlled on finite horizons by the bounded number of LOB jumps together with the Brownian fluctuations of the fundamental value.

\section{Learning Framework}

We model market makers as learning agents trained with multi-agent proximal policy optimization (MAPPO). All experiments were implemented in Python using the RLlib reinforcement-learning library (\cite{liang2018rllib}), while the simulator follows the PettingZoo \cite{terry2021pettingzoo} \texttt{ParallelEnv} interface. In particular, the environment is implemented as a parallel multi-agent market and interfaced with RLlib through the PettingZoo adapter, which yields a standard multi-agent training loop.

The learning setup follows the centralized training with decentralized execution (CTDE) paradigm. During execution, each agent acts only on its local observation. During training, a centralized critic additionally receives a privileged global state and an agent identifier.

\subsection{Observation space}

At each decision step, agent $i$ observes a local state composed of three blocks:  
(i) scalar features,  
(ii) a public summary of the limit order book, and  
(iii) a history window.

The scalar block contains the agent's own inventory and cash, time elapsed measure, market price, last transaction price and any private signal available to informed agents. The public LOB block summarizes liquidity around the mid-price while preserving queue-priority information: instead of collapsing each price level to a single depth number, the observation distinguishes between the agent's own resting liquidity and competing liquidity at tracked price levels. Finally, the history block provides a compact recent trajectory of market and agent-specific variables.

This decomposition was chosen to combine portfolio state, current market conditions, and short-term temporal context, while keeping the actor observation decentralized. The full construction of all observation components is provided in Appendix \ref{app:observation}.

\subsection{Action space}

At each step, the market maker chooses quote placement and quote size independently on the bid and ask sides. For each side, the policy outputs a discrete level from mid to determine the posted price and a discrete size bin.

The size bin determines whether a new order is submitted and, if so, maps deterministically to an explicit order volume. This discretization makes actions easier to interpret economically and reduces the number of behaviorally redundant choices. Order cancellations are not learned directly: they follow deterministic rules designed to stabilize training and to mimic realistic order-management behavior. The exact cancellation mechanism is described in Appendix \ref{app:cancel}.

\subsection{Policy network}

The actor network combines three complementary encoders. First, scalar portfolio and market features are processed by a multilayer perceptron (MLP). Second, the agent's outstanding orders are fused with surrounding market liquidity through a cross-attention block, allowing the policy to condition on queue position and local competitive pressure. Third, a gated recurrent unit (GRU) encodes the history window to provide short-term temporal context. The resulting embeddings are concatenated and passed to the policy head.

To capture regime-dependent behavior, the actor uses a mixture-of-experts (MoE) head. Expert selection is based on the sign of the agent's inventory and, for agents with access to a private value signal, the sign of estimated mispricing $M_t - \widetilde{S}_t$. This design encourages specialization across economically meaningful regimes while keeping execution fully decentralized.

The precise architectural specification is deferred to Appendix \ref{app:agents_net}.

\subsection{Centralized critic}

The centralized critic provides the CTDE value baseline used in MAPPO. It receives a privileged global state containing market-wide variables, per-agent scalar states, and market-wide queue information, together with the identity of the agent whose value is being evaluated. The actor never accesses this privileged information.

Architecturally, the critic mirrors the main design principles of the actor. It combines an MLP-based scalar encoder, which uses a history embedding produced by the same GRU architecture as the actor, with a cross-attention module operating on the market-wide queue state, and maps the resulting representation to a value estimate. As in the actor, the value head uses regime-dependent experts. Here the gating depends on the sign of the current agent's inventory and the sign of the actual mispricing $M_t - S_t$, allowing the critic to distinguish between economically different market states during training.

A full mathematical description of the critic is provided in Appendix \ref{app:critics_net}.

\subsection{Rewards}
\label{sec:rewards}

A central part of the environment specification is the objective of the market makers' training. In the literature, the objective is typically formulated as a terminal wealth criterion penalized by inventory risk; see for instance \cite{AVELLANEDASTOIKOV2008,CARTEAJAIMUNGAL2016,LEHALLENEUMAN2019,NEUMANVOSS2022,BANKCARTEAKORBER2023,BMB2025}.

In the environment, we decompose that reward into terminal and running components. This gives a denser learning signal to the RL algorithm while preserving the same economic trade-off between profitability and inventory control.

At the same time, we further enhance the reward by introducing liquidation prices computed from executable market quotes. This is a natural extension in a fully observable LOB. For brevity, denote $
c_n^{(i)} := c_{t_n^{\mathrm{dec}}}^{(i)}$ and
$q_n^{(i)} := q_{t_n^{\mathrm{dec}}}^{(i)}.
$

For each market maker $i$, let $B_n^{-i}$ and $A_n^{-i}$ denote the best bid and best ask after excluding agent $i$'s own resting orders. We then define the liquidation price by
\begin{equation}
P_{n,\mathrm{liq}}^{(i)}
=
\begin{cases}
B_n^{-i}, & q_n^{(i)} > 0,\\[4pt]
A_n^{-i}, & q_n^{(i)} < 0,\\[4pt]
0, & q_n^{(i)} = 0.
\end{cases}
\label{eq:liq_price}
\end{equation}
Accordingly, the liquidation wealth of agent $i$ at decision time $t_n^{\mathrm{dec}}$ is
\begin{equation}
W_{n,\mathrm{liq}}^{(i)}
=
c_n^{(i)} + q_n^{(i)} P_{n,\mathrm{liq}}^{(i)}.
\label{eq:liq_wealth}
\end{equation}

The one-step reward used by the simulator is the increment in liquidation wealth penalized by inventory holding. Writing
\[
\Delta W_n^{(i)} := W_{n,\mathrm{liq}}^{(i)} - W_{n-1,\mathrm{liq}}^{(i)},
\]
the implemented (unscaled) reward is
\begin{equation}
\tilde r_n^{(i)}
=
\begin{cases}
\Delta W_n^{(i)} - \phi_i \big(q_n^{(i)}\big)^2 \Delta t, & n = 1,\dots,N_{\text{steps}}-1,\\[6pt]
\Delta W_{N_{\text{steps}}}^{(i)} - \alpha_i \big(q_{N_{\text{steps}}}^{(i)}\big)^2, & n = N_{\text{steps}},
\end{cases}
\label{eq:implemented_reward}
\end{equation}
where $N_{\text{steps}}$ is the number of agents' decision steps and $(\alpha_i,\phi_i)$ are agent-specific inventory-aversion parameters. Hence, for all non-terminal decision steps the agent receives a running reward equal to the change in liquidation wealth minus the running inventory penalty, while at the terminal step the running penalty is replaced by a terminal quadratic penalty on the residual inventory.

Finally, to keep the numerical scale of rewards moderate during training, the reward passed to the learning algorithm is normalized by a constant factor. Denoting the unscaled reward in \eqref{eq:implemented_reward} by $\tilde r_n^{(i)}$, the simulator uses the scaled reward
\[
r_n^{(i)}=\frac{1}{N_{\text{steps}}}\,\tilde r_n^{(i)}.
\]
This scaling is introduced purely for numerical stability. This rewards then enters the computations of total reward in PPO algorithm.

\section{Results}

In this section, we analyze how market informedness affects the profitability of market makers. We begin with a stationary setting and then examine the robustness of the results by extending the duration of the experiment and introducing a non-stationary environment.

\subsection{Baseline experiments}

The baseline experiments show that adverse selection is particularly costly when market prices deviate substantially from fundamental value. Within the stationary informedness regime, the overall relationship is upward: market makers' profitability tends to improve as the level of market informedness increases.

Figure \ref{fig:low_vol_low_inf} presents the evaluation results for a single episode consisting of $100$ simulation steps over the time horizon $T=10$, under conditions of low market volatility and low market informedness. Subfigure \ref{subfig:a_low_vol_low_inf} illustrates the evolution of the fundamental value together with the market mid-price and transaction prices. During approximately the first half of the experiment, the market overprices the asset relative to its fundamental value. Around time $t=4$, the market experiences a shock that drives the price back toward fundamentals, resulting in moderate losses distributed across all agents. Following this reversal, the sign of the mispricing changes, and, through the self-exciting dynamics of the Hawkes process, the market eventually begins to underprice the asset.

Subfigure \ref{subfig:b_low_vol_low_inf} suggests that this price movement is driven primarily by the aggressive policy. Aggressive agents accumulate the largest inventory positions during the trend reversal and subsequently unwind them at less favorable prices during the period of market stagnation between $t=6$ and $t=8$. As a result, they incur substantial losses, which are shown in Subfigure \ref{subfig:c_low_vol_low_inf}.

\begin{figure}[htbp]
    \centering
    \begin{subfigure}[t]{0.3\textwidth}
        \centering
        \includegraphics[width=\textwidth]{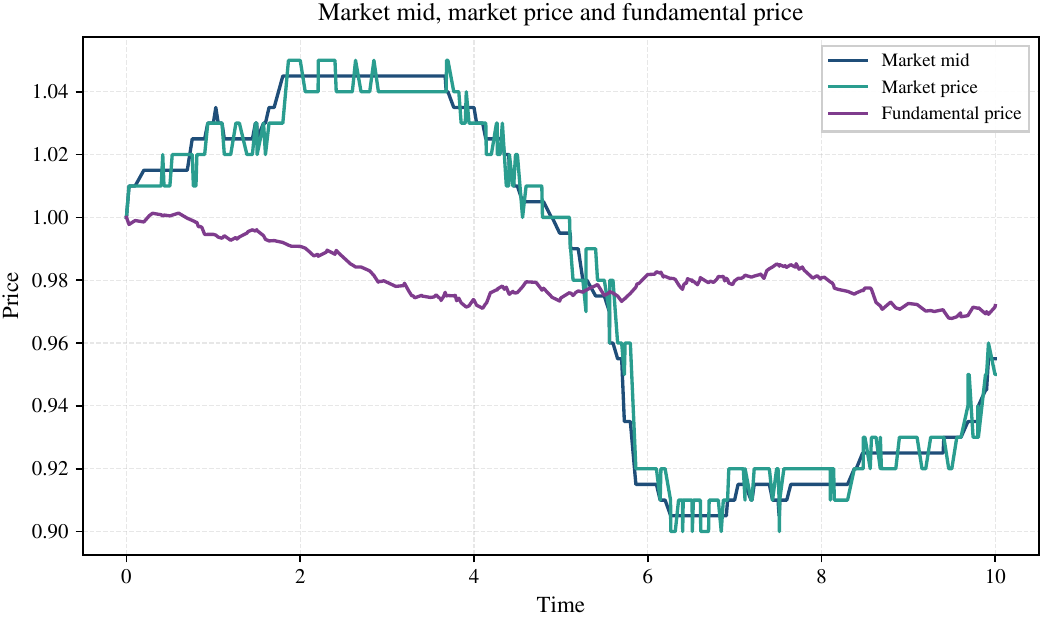}
        \caption{Price evolution.}
        \label{subfig:a_low_vol_low_inf}
    \end{subfigure}
    \hfill
    \begin{subfigure}[t]{0.3\textwidth}
        \centering
        \includegraphics[width=\textwidth]{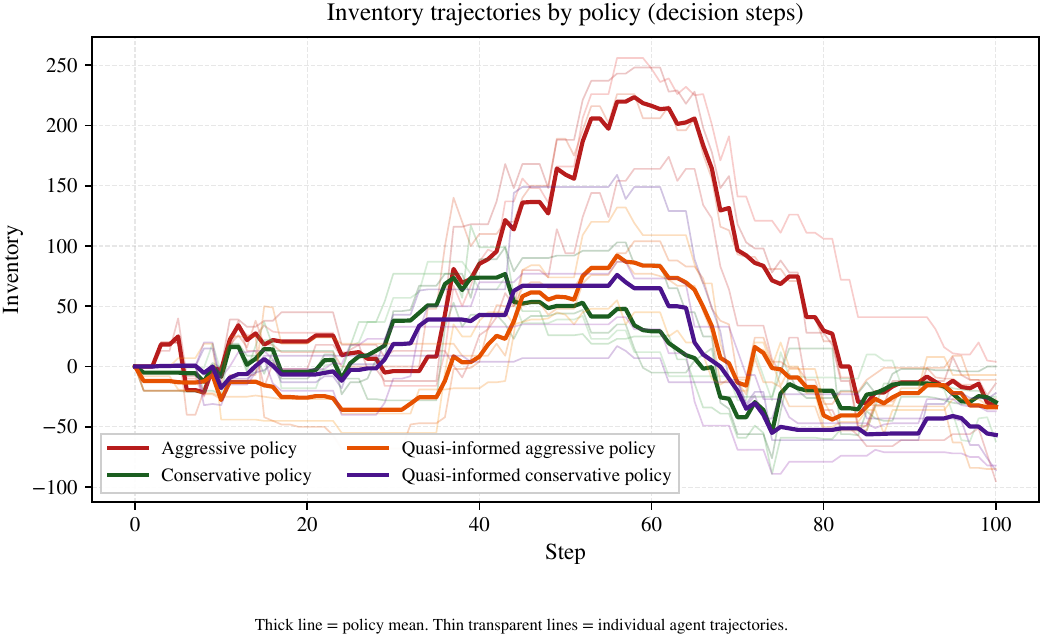}
        \caption{Inventory evolution.}
        \label{subfig:b_low_vol_low_inf}
    \end{subfigure}
    \begin{subfigure}[t]{0.3\textwidth}
        \centering
        \includegraphics[width=\textwidth]{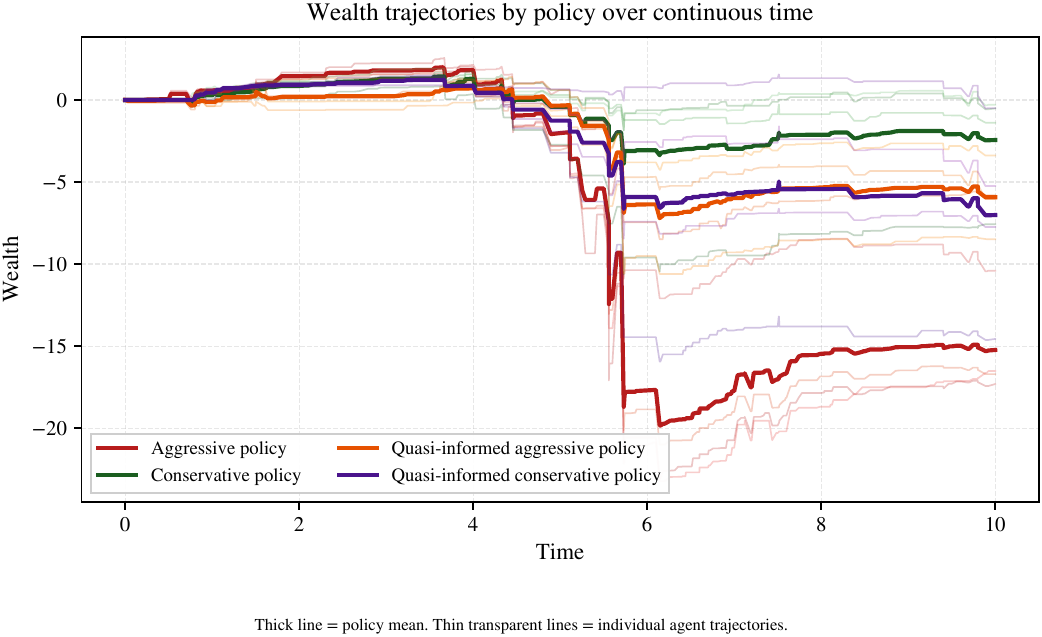}
        \caption{Wealth evolution.}
        \label{subfig:c_low_vol_low_inf}
    \end{subfigure}
    
    \caption{Simulation under low volatility ($\sigma = 0.005$) and low informedness ($\psi = 0.01$). Seed = $42$.}
    
    \label{fig:low_vol_low_inf}
\end{figure}

Figure \ref{fig:low_vol_high_inf} illustrates the results of a single episode conducted under the same setup, with the only difference being a higher level of market informedness, set to $\psi = 0.95$. Under this specification, market price dynamics differ markedly from those observed in the low-informedness case. As shown in Subfigure \ref{subfig:a_low_vol_high_inf}, mispricing remains relatively limited, typically amounting to only $1$--$2$ ticks and often falling below the tick size. Subfigure \ref{subfig:b_low_vol_high_inf} shows that agents maintain relatively low inventory levels, reflecting the greater predictability of market movements. This more effective inventory management translates into a controlled and steady increase in terminal wealth across all policies.

\begin{figure}[htbp]
    \centering
    \begin{subfigure}[t]{0.3\textwidth}
        \centering
        \includegraphics[width=\textwidth]{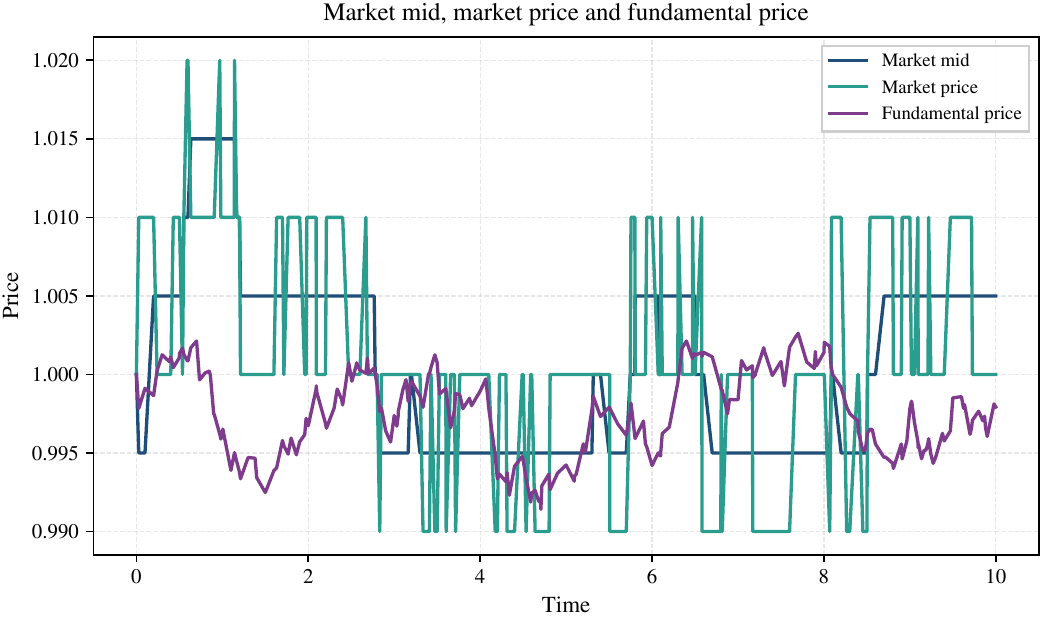}
        \caption{Price evolution.}
        \label{subfig:a_low_vol_high_inf}
    \end{subfigure}
    \hfill
    \begin{subfigure}[t]{0.3\textwidth}
        \centering
        \includegraphics[width=\textwidth]{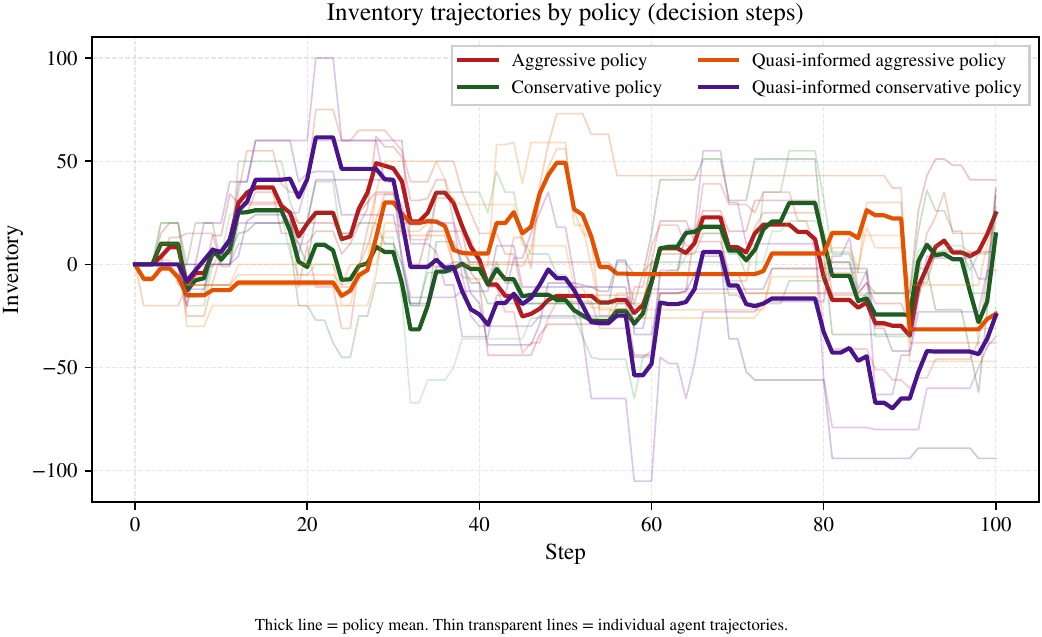}
        \caption{Inventory evolution.}
        \label{subfig:b_low_vol_high_inf}
    \end{subfigure}
    \begin{subfigure}[t]{0.3\textwidth}
        \centering
        \includegraphics[width=\textwidth]{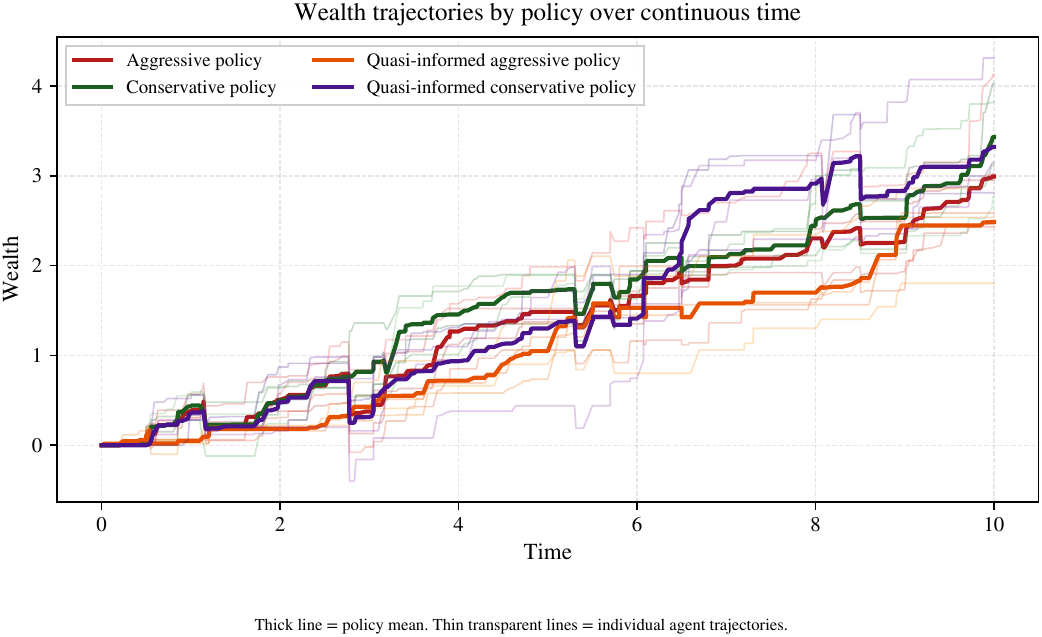}
        \caption{Wealth evolution.}
        \label{subfig:c_low_vol_high_inf}
    \end{subfigure}
    
    \caption{Simulation under low volatility ($\sigma = 0.005$) and high informedness ($\psi = 0.95$). Seed = $42$.}
    
    \label{fig:low_vol_high_inf}
\end{figure}

These two experiments provide qualitative evidence that low market informedness is generally unfavorable for market makers, as it limits their ability to predict future price movements on the basis of past market states. In such an environment, the relatively small informed segment of the market generates adverse selection, which may trigger substantial reversals in market regimes and, consequently, losses for market makers who are unable to unwind their positions sufficiently quickly.

This interpretation is supported by Table \ref{tab:wealth_combined_vol}, which reports the average terminal wealth for each policy and each informedness ratio $\psi$. For the terminal wealth associated with each policy, as well as for the aggregate terminal wealth corresponding to each value of $\psi$, an ordinary least squares (OLS) regression has been estimated in order to identify the general trend. In all cases, the estimated slope coefficient is positive, which is consistent with the conclusions drawn from the single-episode analysis. Repeating the same experiments under higher market volatility, $\sigma = 0.02$, leads to the same qualitative conclusion: market makers' profitability exhibits an overall upward trend as a function of market informedness.

\begin{table}[!htbp]
\centering
\scriptsize
\setlength{\tabcolsep}{1.8pt}
\renewcommand{\arraystretch}{1.15}
\resizebox{\textwidth}{!}{%
\begin{tabular}{lccccccccccccc}
\toprule
Policy & \rotatebox{90}{$\psi=0,01$} & \rotatebox{90}{$\psi=0,05$} & \rotatebox{90}{$\psi=0,10$} & \rotatebox{90}{$\psi=0,30$} & \rotatebox{90}{$\psi=0,50$} & \rotatebox{90}{$\psi=0,70$} & \rotatebox{90}{$\psi=0,90$} & \rotatebox{90}{$\psi=0,95$} & \rotatebox{90}{$\psi=0,99$} & \rotatebox{90}{$\psi=1,00$} & \rotatebox{90}{$\hat\alpha$} & \rotatebox{90}{$\hat\beta$} & \rotatebox{90}{$R^2$} \\
\midrule

\multicolumn{14}{l}{\textit{Low volatility} ($\sigma = 0.005$)} \\
\midrule
Aggressive & \makecell{-0,73\\$\pm$ 5,95} & \makecell{1,68\\$\pm$ 2,57} & \makecell{1,56\\$\pm$ 2,10} & \makecell{2,88\\$\pm$ 1,54} & \makecell{4,36\\$\pm$ 1,60} & \makecell{2,90\\$\pm$ 1,11} & \makecell{4,62\\$\pm$ 1,12} & \makecell{3,61\\$\pm$ 0,98} & \makecell{3,38\\$\pm$ 0,97} & \makecell{3,93\\$\pm$ 1,23} & 1,12 & 3,09 & 0,62 \\
Conservative & \makecell{0,96\\$\pm$ 3,40} & \makecell{1,25\\$\pm$ 2,58} & \makecell{2,44\\$\pm$ 1,27} & \makecell{3,05\\$\pm$ 1,32} & \makecell{1,96\\$\pm$ 0,83} & \makecell{2,56\\$\pm$ 1,12} & \makecell{3,06\\$\pm$ 0,77} & \makecell{3,34\\$\pm$ 0,84} & \makecell{3,40\\$\pm$ 0,90} & \makecell{3,09\\$\pm$ 0,97} & 1,57 & 1,71 & 0,65 \\
\makecell[l]{Quasi-informed\\aggressive} & \makecell{-0,50\\$\pm$ 3,24} & \makecell{0,85\\$\pm$ 3,40} & \makecell{3,04\\$\pm$ 1,62} & \makecell{2,58\\$\pm$ 1,01} & \makecell{3,43\\$\pm$ 1,27} & \makecell{3,53\\$\pm$ 0,89} & \makecell{4,70\\$\pm$ 1,39} & \makecell{3,58\\$\pm$ 0,97} & \makecell{4,11\\$\pm$ 1,00} & \makecell{3,20\\$\pm$ 1,04} & 1,21 & 2,99 & 0,62 \\
\makecell[l]{Quasi-informed\\conservative} & \makecell{-0,91\\$\pm$ 3,16} & \makecell{0,71\\$\pm$ 1,79} & \makecell{2,88\\$\pm$ 1,78} & \makecell{2,55\\$\pm$ 1,38} & \makecell{2,91\\$\pm$ 1,06} & \makecell{3,40\\$\pm$ 1,34} & \makecell{2,54\\$\pm$ 0,69} & \makecell{2,80\\$\pm$ 0,84} & \makecell{3,14\\$\pm$ 0,90} & \makecell{2,24\\$\pm$ 0,75} & 1,15 & 1,96 & 0,37 \\
\textbf{Sum} & \textbf{-1,18} & \textbf{4,49} & \textbf{9,92} & \textbf{11,06} & \textbf{12,66} & \textbf{12,39} & \textbf{14,92} & \textbf{13,33} & \textbf{14,03} & \textbf{12,46} & \textbf{5,05} & \textbf{9,75} & \textbf{0,64} \\

\midrule
\multicolumn{14}{l}{\textit{High volatility} ($\sigma = 0.02$)} \\
\midrule
Aggressive & \makecell{-3,12\\$\pm$ 6,14} & \makecell{-1,47\\$\pm$ 2,76} & \makecell{0,60\\$\pm$ 3,11} & \makecell{-0,08\\$\pm$ 3,99} & \makecell{0,19\\$\pm$ 3,22} & \makecell{-0,07\\$\pm$ 3,23} & \makecell{0,51\\$\pm$ 3,27} & \makecell{2,36\\$\pm$ 2,93} & \makecell{0,44\\$\pm$ 4,43} & \makecell{2,73\\$\pm$ 5,01} & -1,50 & 3,11 & 0,58 \\
Conservative & \makecell{-2,13\\$\pm$ 6,20} & \makecell{-0,26\\$\pm$ 2,96} & \makecell{0,65\\$\pm$ 2,26} & \makecell{0,66\\$\pm$ 1,66} & \makecell{0,09\\$\pm$ 2,18} & \makecell{0,68\\$\pm$ 2,26} & \makecell{0,68\\$\pm$ 2,96} & \makecell{2,36\\$\pm$ 2,05} & \makecell{1,04\\$\pm$ 2,20} & \makecell{1,89\\$\pm$ 2,75} & -0,68 & 2,26 & 0,57 \\
\makecell[l]{Quasi-informed\\aggressive} & \makecell{-0,26\\$\pm$ 4,41} & \makecell{2,67\\$\pm$ 3,59} & \makecell{6,33\\$\pm$ 3,55} & \makecell{4,27\\$\pm$ 2,39} & \makecell{2,11\\$\pm$ 2,79} & \makecell{2,58\\$\pm$ 2,46} & \makecell{2,39\\$\pm$ 2,30} & \makecell{3,67\\$\pm$ 2,66} & \makecell{1,77\\$\pm$ 3,16} & \makecell{4,54\\$\pm$ 2,13} & 2,90 & 0,20 & 0,00 \\
\makecell[l]{Quasi-informed\\conservative} & \makecell{-2,44\\$\pm$ 4,61} & \makecell{0,11\\$\pm$ 4,68} & \makecell{-0,82\\$\pm$ 4,18} & \makecell{1,62\\$\pm$ 3,25} & \makecell{2,54\\$\pm$ 2,31} & \makecell{-0,48\\$\pm$ 2,42} & \makecell{0,41\\$\pm$ 2,34} & \makecell{3,08\\$\pm$ 1,98} & \makecell{-0,56\\$\pm$ 2,37} & \makecell{3,66\\$\pm$ 2,43} & -0,68 & 2,52 & 0,28 \\
\textbf{Sum} & \textbf{-7,95} & \textbf{1,05} & \textbf{6,76} & \textbf{6,47} & \textbf{4,93} & \textbf{2,71} & \textbf{3,99} & \textbf{11,47} & \textbf{2,69} & \textbf{12,82} & \textbf{0,05} & \textbf{8,09} & \textbf{0,33} \\
\bottomrule
\end{tabular}}
\caption{Final wealth distribution across informed ratios $\psi$ under low volatility ($\sigma = 0.005$) and high volatility ($\sigma = 0.02$).
Policy-level entries are mean $\pm$ standard deviation (STD), rounded to two decimals.
The \textbf{Sum} row reports the sum of policy means.
The last three columns report $(\hat\alpha,\hat\beta,R^2)$ from the OLS regression of mean final wealth on the informed ratio $\psi$.
Training for 50 iterations evaluated on 400 distinct episodes.}
\label{tab:wealth_combined_vol}
\end{table}

\subsection{Quoting}

Quoting behavior differs systematically across policies, but it does not display a clear monotonic relationship with market informedness. The evidence instead suggests that profitability is driven primarily by overall market conditions, even though agents learn to adapt their quoting decisions to inventory positions and mispricing regimes.

Another important aspect of the baseline analysis concerns quoting behavior and its dependence on policy type and market informedness.

Averaged across all informedness ratios $\psi$, posted spreads differ across policies in an economically intuitive way. In the low-volatility regime, $\sigma = 0.005$, the Aggressive policy posts the tightest quotes on average, with a mean spread of $2.41$ ticks, whereas the Conservative policy posts an average spread of $2.64$ ticks, approximately $9.5\%$ higher. The quasi-informed policies quote more widely still, with average spreads of $2.79$ and $2.80$ ticks for the aggressive and conservative variants, respectively. This cross-policy ordering is consistent with the intended distinction between more and less aggressive quoting styles.

At the same time, the relationship between the average posted spread and the informedness ratio is weak and irregular. Regressing average spread on $\psi$ yields very low explanatory power for all policies, with $R^2 < 0.05$ throughout. This suggests that the previously documented increase in terminal wealth with market informedness is not driven by a simple monotonic adjustment in average quoting width. Rather, the profitability effect appears to stem primarily from changes in the market environment itself, such as the reduction in persistent mispricing and adverse selection at higher values of $\psi$.

A similar conclusion follows from the analysis of posted volume across spread levels. Across policies, agents generally post larger volumes closer to the mid-price and smaller volumes further away. The decline is moderate: mean posted volume typically decreases by about one unit between the closest and farthest quoted levels, that is, from slightly above $3$ to slightly above $2$. Aggressive policies also tend to post somewhat larger volumes near the mid-price than their conservative counterparts.

However, as in the case of average spread, posted volume at a given level from mid generally exhibits no strong linear relation with $\psi$, and in most cases the regression fit remains weak, with $R^2 < 0.10$. Two exceptions are the Quasi-informed aggressive policy at level 1, for which $R^2 = 0.484$ with a positive slope, and the Conservative policy at level 0, for which $R^2 = 0.30$ with a negative slope. Aside from these isolated cases, no common economically meaningful pattern is visible across policies and levels. Taken together, these findings support the view that the increase in profitability associated with higher market informedness is driven primarily by improvements in the market environment rather than by systematic changes in average quoting width or posted size.

At first glance, these results might suggest that the policies do not learn meaningful behavior. However, an analysis conditional on inventory and mispricing regimes indicates otherwise. Averaged across all informedness ratios $\psi$, and for both low and high market volatility, all policies tend to post tighter quotes and larger volumes on the ask side when inventory is positive, and analogously on the bid side when inventory is negative. Quantitatively, the relevant half-spreads are typically narrower by a factor of approximately $1.6$--$3.2$, while the corresponding posted volumes are larger by about $35\%$--$120\%$, with an average increase of roughly $70\%$. This pattern is consistent with meaningful inventory-aware specialization.

Moreover, the uninformed policies remain nearly symmetric across underpricing and overpricing regimes once inventory is held fixed, which is expected given that they do not observe the fundamental value. By contrast, the quasi-informed policies do differentiate their behavior across mispricing regimes. This distinction is especially clear for the Quasi-informed aggressive policy and less regular for the Quasi-informed conservative one. In particular, both policies tend to post higher ask-side volumes in overpricing regimes and higher bid-side volumes in underpricing regimes, suggesting that informed agents attempt to unwind positions when prices are favorable. When this pattern is analyzed separately for each value of $\psi$, it remains persistent for the Quasi-informed aggressive policy, but becomes more ambiguous for its conservative counterpart. This suggests that the Quasi-informed conservative policy faces a stronger trade-off between exploiting informational advantages and controlling inventory risk.

No equally clear pattern emerges for bid and ask half-spreads. This is likely due to the role of stochastic market order arrivals and random queueing in determining the effective importance of quoted spreads. Because market orders arriving between decision steps may consume liquidity beyond the first level, posted volume appears to be more consequential than the precise choice of quoting level, making the effect of half-spread adjustments less distinct.

To further assess the robustness of these findings and verify that the learned policies are not merely overfitting the specific market structure considered so far, we extend the analysis in two directions. First, we evaluate the already trained policies in substantially longer episodes than those used during training. Second, we introduce non-stationarity into the market regime in order to examine whether the main conclusions remain valid in a more dynamic environment.

\subsection{Expanded experiment}

An important robustness question is whether policies trained on episodes of 100 decision steps retain their performance in substantially longer runs. The expanded experiment addresses this issue by evaluating the same policies in episodes of 600 decision steps.

The results of the extended experiment are presented in Figure \ref{fig:long_low_vol_low_inf}. In this specification, market informedness is set to $\psi = 0.10$. In such an environment, mispricings occur frequently and often reach a magnitude of several ticks. The resulting trend reversals are qualitatively similar to those observed in the shorter experiments conducted in less informed markets.

The analysis of inventory dynamics shows that inventories, particularly for the Aggressive policy, can become large in absolute value. However, over a longer horizon they continue to fluctuate around zero rather than drift persistently in one direction. The wealth trajectories display a generally positive trend over time, interrupted by occasional sharp declines, again most visibly for the Aggressive policy. Overall, the behavior of the policies appears stable throughout the full episode. In particular, there is no visible change in behavior after time $t=10$ (or step 100), which corresponds to the original time horizon used during training. This indicates that the learned policies generalize well to longer experiments. Apart from scale, no clear quantitative change is visible in the effect of market informedness on market makers' profitability in this longer-horizon setting.

\begin{figure}[H]
    \centering
    \begin{subfigure}[t]{0.3\textwidth}
        \centering
        \includegraphics[width=\textwidth]{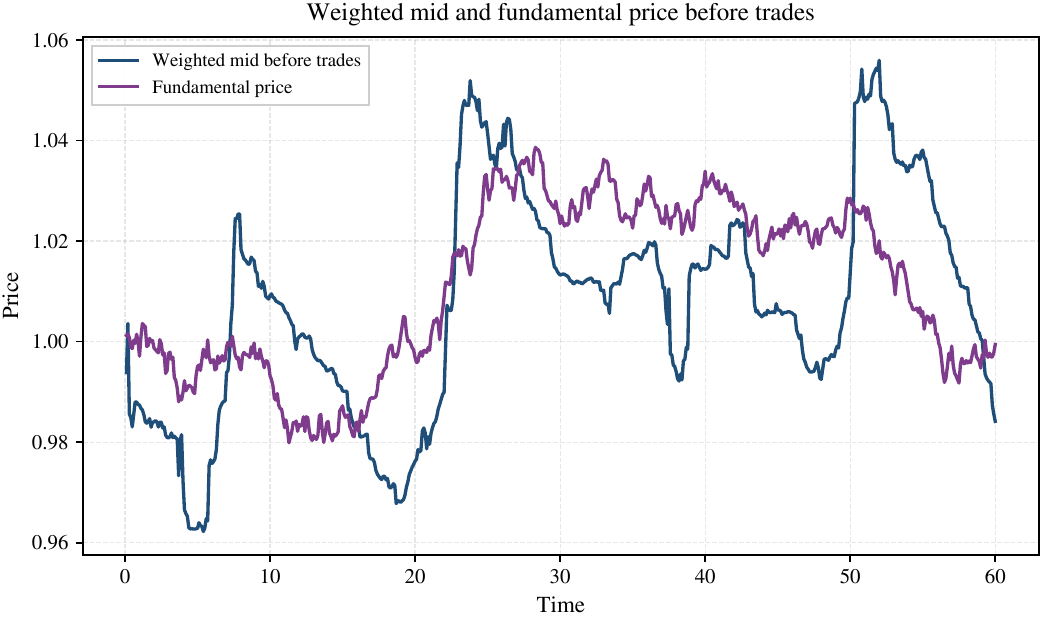}
        \caption{Weighted market mid and fundamental value.}
        \label{subfig:long_price}
    \end{subfigure}
    \hfill
    \begin{subfigure}[t]{0.3\textwidth}
        \centering
        \includegraphics[width=\textwidth]{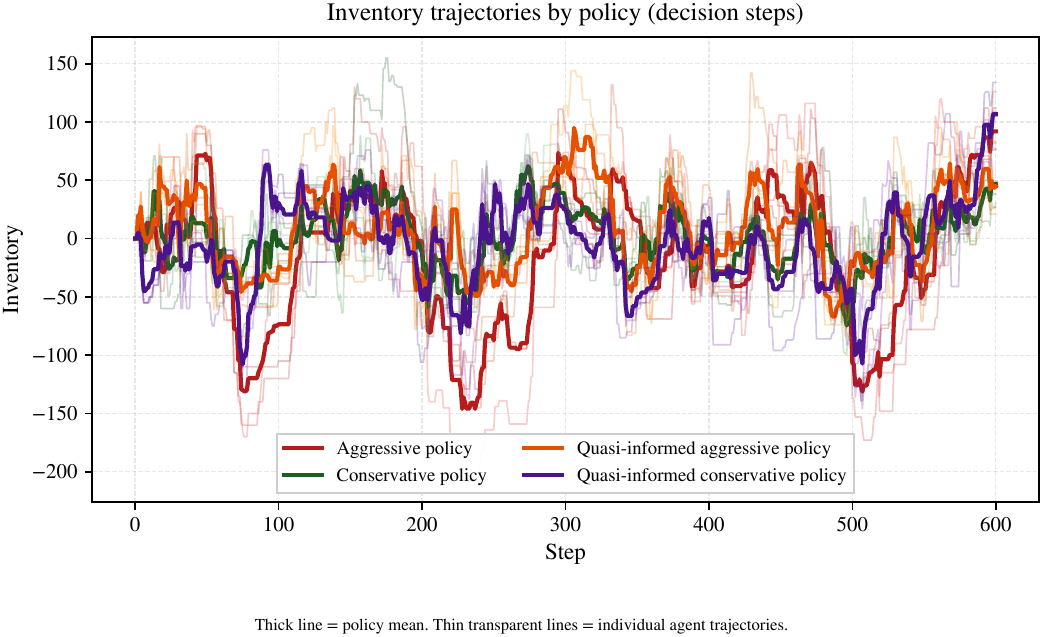}
        \caption{Inventory evolution.}
        \label{subfig:inv_long}
    \end{subfigure}
    \begin{subfigure}[t]{0.3\textwidth}
        \centering
        \includegraphics[width=\textwidth]{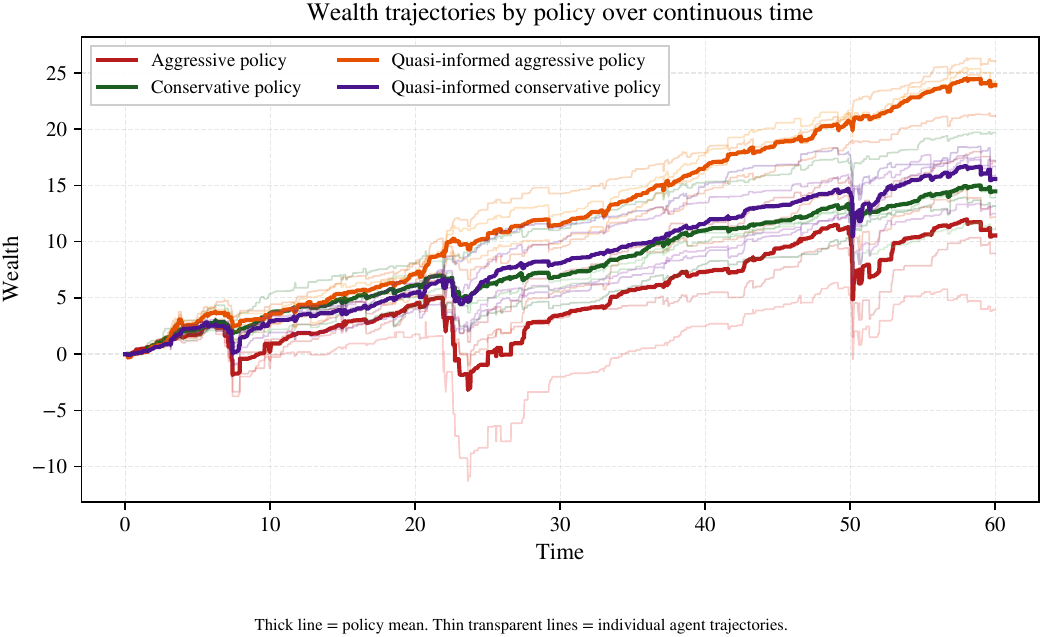}
        \caption{Wealth evolution.}
        \label{subfig:long_wealth}
    \end{subfigure}
    
    \caption{Simulation under low volatility ($\sigma = 0.005$) and medium-low informedness ($\psi = 0.10$). Seed = $37$.}
    
    \label{fig:long_low_vol_low_inf}
\end{figure}

Quantitatively, the results remain similar to those of the baseline experiment.

\subsection{Non-stationary experiments}
The non-stationary experiments indicate that the positive relationship between market informedness and agents' profitability is not merely an artifact of fitting a stationary environment. Instead, it remains visible across a broad range of non-stationary market regimes.

While the expanded-horizon experiment shows that the learned policies generalize over time, an additional question is whether they also remain effective when the underlying market environment itself changes. To address this issue, we consider a range of non-stationary regime specifications for the informed ratio $\psi$ and market volatility $\sigma$, covering both cases in which agents are aware of the regime and cases in which they are not. The results are reported in Table \ref{tab:mean_wealth_change_by_psi}, which presents the mean per-step wealth change for each policy and informedness level, together with OLS estimates summarizing the relationship between profitability and $\psi$ under each regime specification.

Overall, the evidence is consistent with the main conclusion of the stationary analysis. For every regime specification, the aggregate measure of profitability, reported in the \textbf{Sum} row, exhibits a positive slope coefficient with respect to market informedness. This indicates that the previously documented upward trend in market makers' wealth is not driven solely by the stationary benchmark, but persists even when the market environment changes over time. In several cases, this relationship is also quantitatively strong, as reflected in relatively high values of $R^2$, particularly for regime models 4, 5, and 6. This suggests that the positive association between informedness and profitability remains stable across a broad set of structured regime-switching environments. 

At the policy level, the same pattern largely continues to hold. The Aggressive, Quasi-informed aggressive, and Quasi-informed conservative policies display positive slope coefficients in all reported regime specifications. The only systematic exception is the Conservative policy, for which the estimated slope becomes negative in regime models 5 and 7. This suggests that the most inventory-protective policy may be less capable of benefiting from higher market informedness in environments with more frequent or irregular regime changes. By contrast, the quasi-informed policies remain positively related to $\psi$ throughout, which is consistent with the broader interpretation of this study: access to informational signals becomes increasingly valuable as market informedness rises, even when the surrounding market structure is non-stationary.

At the same time, the table also shows that the strength of this relationship depends on the type of non-stationarity considered. In particular, the most irregular specifications, especially regime model 7 in the unknown setting, produce a much weaker fit. Thus, non-stationarity does not eliminate the positive profitability effect of informedness, but it does make the relationship noisier and less uniform, possibly due to limitations of machine learning. Taken together, these results support the view that the main finding of the paper is robust: higher market informedness is generally associated with better market-making performance, and this conclusion survives both temporal extension of the experiment and substantial departures from stationarity.
\FloatBarrier

\begin{table}[ht]
\centering
\scriptsize
\setlength{\tabcolsep}{2.5pt}
\renewcommand{\arraystretch}{1.1}
\resizebox{\textwidth}{!}{%
\begin{tabular}{ll*{15}{r}rrr}
\toprule
\multirow{2}{*}{Regime mode} & \multirow{2}{*}{Policy} & \multicolumn{15}{c}{$\psi$ (informed ratio)} & \multicolumn{3}{c}{OLS on $\psi$} \\
\cmidrule(lr){3-17}\cmidrule(lr){18-20}
 &  & 0.00 & 0.01 & 0.05 & 0.10 & 0.20 & 0.30 & 0.40 & 0.50 & 0.60 & 0.70 & 0.80 & 0.90 & 0.95 & 0.99 & 1.00 & $\hat\beta$ & $\hat\alpha$ & $R^2$ \\
\midrule
\multirow{5}{*}{\texttt{Regime model 2} (known)} & \texttt{agg} & -0.64 & -0.19 & -0.04 & 0.04 & 0.15 & 0.03 & 0.26 & 0.26 & 0.24 & 0.30 & 0.29 & 0.34 & 0.31 & 0.26 & 0.33 & 0.54 & -0.14 & 0.59 \\
 & \texttt{cons} & -0.01 & -0.05 & 0.03 & 0.00 & 0.07 & 0.03 & 0.14 & 0.13 & 0.13 & 0.13 & 0.09 & 0.12 & 0.16 & 0.09 & 0.17 & 0.15 & 0.01 & 0.67 \\
 & \texttt{quasi-inf-agg} & -0.67 & 0.02 & 0.14 & 0.19 & 0.28 & 0.22 & 0.33 & 0.35 & 0.30 & 0.31 & 0.35 & 0.34 & 0.36 & 0.32 & 0.36 & 0.44 & -0.00 & 0.39 \\
 & \texttt{quasi-inf-cons} & -1.36 & 0.02 & 0.16 & 0.23 & 0.33 & 0.17 & 0.34 & 0.33 & 0.30 & 0.39 & 0.36 & 0.38 & 0.37 & 0.28 & 0.39 & 0.61 & -0.12 & 0.27 \\
 & \textbf{Sum} & \textbf{-2.67} & \textbf{-0.19} & \textbf{0.29} & \textbf{0.47} & \textbf{0.84} & \textbf{0.44} & \textbf{1.07} & \textbf{1.06} & \textbf{0.97} & \textbf{1.33} & \textbf{1.09} & \textbf{1.18} & \textbf{1.20} & \textbf{0.95} & \textbf{1.24} & \textbf{1.74} & \textbf{-0.25} & \textbf{0.40} \\
\addlinespace[0.5ex]
\multirow{5}{*}{\texttt{Regime model 2} (unknown)} & \texttt{agg} & -0.52 & -0.14 & -0.03 & -0.02 & 0.10 & -0.05 & 0.01 & 0.12 & 0.09 & 0.23 & 0.18 & 0.25 & 0.21 & 0.35 & 0.25 & 0.47 & -0.17 & 0.69 \\
 & \texttt{cons} & -0.20 & 0.16 & 0.13 & 0.24 & 0.18 & 0.07 & 0.13 & 0.26 & 0.22 & 0.26 & 0.20 & 0.23 & 0.25 & 0.23 & 0.27 & 0.19 & 0.08 & 0.35 \\
 & \texttt{quasi-inf-agg} & -0.90 & 0.10 & 0.23 & 0.30 & 0.40 & 0.22 & 0.27 & 0.40 & 0.38 & 0.39 & 0.31 & 0.43 & 0.37 & 0.41 & 0.42 & 0.47 & 0.01 & 0.29 \\
 & \texttt{quasi-inf-cons} & -0.49 & 0.14 & 0.25 & 0.23 & 0.30 & 0.16 & 0.23 & 0.31 & 0.30 & 0.32 & 0.26 & 0.30 & 0.30 & 0.29 & 0.30 & 0.27 & 0.07 & 0.26 \\
 & \textbf{Sum} & \textbf{-2.11} & \textbf{0.26} & \textbf{0.58} & \textbf{0.76} & \textbf{0.97} & \textbf{0.39} & \textbf{0.63} & \textbf{1.09} & \textbf{0.99} & \textbf{1.20} & \textbf{0.95} & \textbf{1.21} & \textbf{1.13} & \textbf{1.29} & \textbf{1.24} & \textbf{1.40} & \textbf{-0.00} & \textbf{0.39} \\
\addlinespace[0.5ex]
\multirow{5}{*}{\texttt{Regime model 4} (known)} & \texttt{agg} & -0.38 & -0.16 & -0.11 & -0.12 & 0.02 & 0.04 & 0.05 & 0.00 & 0.08 & 0.07 & 0.07 & 0.17 & 0.18 & 0.19 & 0.12 & 0.29 & -0.13 & 0.65 \\
 & \texttt{cons} & -0.03 & -0.10 & -0.06 & -0.10 & -0.01 & 0.02 & 0.02 & 0.08 & 0.09 & 0.03 & -0.02 & 0.04 & 0.01 & 0.09 & 0.08 & 0.12 & -0.05 & 0.58 \\
 & \texttt{quasi-inf-agg} & -0.85 & -0.05 & 0.06 & 0.16 & 0.19 & 0.16 & 0.22 & 0.26 & 0.24 & 0.27 & 0.24 & 0.32 & 0.30 & 0.32 & 0.31 & 0.45 & -0.10 & 0.64 \\
 & \texttt{quasi-inf-cons} & -0.74 & -0.00 & 0.07 & 0.15 & 0.15 & 0.11 & 0.16 & 0.20 & 0.22 & 0.19 & 0.12 & 0.30 & 0.27 & 0.29 & 0.22 & 0.41 & -0.10 & 0.62 \\
 & \textbf{Sum} & \textbf{-2.00} & \textbf{-0.31} & \textbf{-0.04} & \textbf{0.09} & \textbf{0.35} & \textbf{0.33} & \textbf{0.45} & \textbf{0.55} & \textbf{0.63} & \textbf{0.56} & \textbf{0.41} & \textbf{0.82} & \textbf{0.76} & \textbf{0.89} & \textbf{0.73} & \textbf{1.27} & \textbf{-0.38} & \textbf{0.70} \\
\addlinespace[0.5ex]
\multirow{5}{*}{\texttt{Regime model 4} (unknown)} & \texttt{agg} & -0.34 & -0.11 & -0.10 & -0.06 & 0.04 & 0.05 & 0.04 & 0.06 & 0.12 & 0.08 & 0.04 & 0.13 & 0.17 & 0.14 & 0.14 & 0.25 & -0.12 & 0.70 \\
 & \texttt{cons} & 0.02 & -0.02 & 0.06 & 0.02 & 0.10 & 0.04 & 0.06 & 0.09 & 0.11 & 0.08 & 0.04 & 0.10 & 0.17 & 0.12 & 0.15 & 0.14 & -0.01 & 0.74 \\
 & \texttt{quasi-inf-agg} & -0.42 & 0.04 & 0.08 & 0.14 & 0.19 & 0.14 & 0.18 & 0.25 & 0.25 & 0.22 & 0.15 & 0.26 & 0.31 & 0.30 & 0.30 & 0.34 & -0.02 & 0.79 \\
 & \texttt{quasi-inf-cons} & -0.43 & 0.06 & 0.10 & 0.16 & 0.21 & 0.15 & 0.16 & 0.25 & 0.25 & 0.23 & 0.18 & 0.26 & 0.30 & 0.28 & 0.30 & 0.33 & -0.01 & 0.77 \\
 & \textbf{Sum} & \textbf{-1.17} & \textbf{-0.03} & \textbf{0.13} & \textbf{0.26} & \textbf{0.54} & \textbf{0.37} & \textbf{0.44} & \textbf{0.65} & \textbf{0.73} & \textbf{0.61} & \textbf{0.41} & \textbf{0.75} & \textbf{0.96} & \textbf{0.84} & \textbf{0.89} & \textbf{1.05} & \textbf{-0.15} & \textbf{0.79} \\
\addlinespace[0.5ex]
\multirow{5}{*}{\texttt{Regime model 5} (known)} & \texttt{agg} & -0.38 & -0.32 & -0.09 & -0.02 & 0.03 & 0.02 & 0.09 & 0.09 & 0.13 & 0.11 & 0.06 & 0.02 & 0.17 & 0.14 & 0.03 & 0.28 & -0.06 & 0.49 \\
 & \texttt{cons} & 0.17 & 0.10 & 0.09 & 0.09 & 0.08 & 0.08 & 0.05 & 0.06 & 0.03 & 0.04 & -0.03 & -0.05 & -0.06 & 0.00 & -0.00 & -0.12 & 0.12 & 0.57 \\
 & \texttt{quasi-inf-agg} & -0.20 & -0.04 & 0.08 & 0.12 & 0.14 & 0.14 & 0.20 & 0.21 & 0.27 & 0.32 & 0.25 & 0.23 & 0.27 & 0.28 & 0.20 & 0.34 & 0.02 & 0.84 \\
 & \texttt{quasi-inf-cons} & -0.15 & -0.03 & 0.08 & 0.11 & 0.17 & 0.16 & 0.21 & 0.20 & 0.23 & 0.21 & 0.19 & 0.09 & 0.06 & 0.13 & 0.08 & 0.26 & 0.05 & 0.74 \\
 & \textbf{Sum} & \textbf{-0.55} & \textbf{-0.29} & \textbf{0.16} & \textbf{0.30} & \textbf{0.42} & \textbf{0.40} & \textbf{0.56} & \textbf{0.57} & \textbf{0.67} & \textbf{0.69} & \textbf{0.46} & \textbf{0.29} & \textbf{0.44} & \textbf{0.55} & \textbf{0.31} & \textbf{0.76} & \textbf{0.13} & \textbf{0.73} \\
\addlinespace[0.5ex]
\multirow{5}{*}{\texttt{Regime model 5} (unknown)} & \texttt{agg} & -0.32 & -0.22 & -0.09 & -0.05 & 0.01 & 0.05 & 0.08 & 0.06 & 0.10 & 0.13 & 0.12 & 0.06 & 0.12 & 0.13 & 0.06 & 0.25 & -0.11 & 0.78 \\
 & \texttt{cons} & 0.14 & 0.04 & 0.06 & 0.03 & 0.04 & 0.05 & 0.04 & 0.01 & 0.06 & 0.05 & -0.02 & -0.06 & -0.05 & 0.02 & -0.01 & -0.10 & 0.09 & 0.65 \\
 & \texttt{quasi-inf-agg} & -0.29 & 0.01 & 0.06 & 0.09 & 0.12 & 0.15 & 0.21 & 0.18 & 0.24 & 0.28 & 0.25 & 0.18 & 0.22 & 0.23 & 0.19 & 0.33 & 0.00 & 0.88 \\
 & \texttt{quasi-inf-cons} & -0.20 & -0.01 & 0.06 & 0.08 & 0.11 & 0.13 & 0.18 & 0.16 & 0.20 & 0.23 & 0.21 & 0.15 & 0.15 & 0.17 & 0.11 & 0.28 & 0.01 & 0.88 \\
 & \textbf{Sum} & \textbf{-0.67} & \textbf{-0.18} & \textbf{0.09} & \textbf{0.15} & \textbf{0.28} & \textbf{0.37} & \textbf{0.52} & \textbf{0.41} & \textbf{0.60} & \textbf{0.69} & \textbf{0.56} & \textbf{0.33} & \textbf{0.45} & \textbf{0.55} & \textbf{0.35} & \textbf{0.75} & \textbf{-0.01} & \textbf{0.92} \\
\addlinespace[0.5ex]
\multirow{5}{*}{\texttt{Regime model 6} (known)} & \texttt{agg} & -0.31 & -0.17 & -0.11 & -0.10 & 0.03 & 0.06 & 0.04 & 0.05 & 0.11 & 0.11 & 0.10 & 0.17 & 0.19 & 0.19 & 0.15 & 0.29 & -0.11 & 0.76 \\
 & \texttt{cons} & -0.04 & -0.08 & -0.04 & -0.05 & 0.06 & 0.04 & 0.02 & 0.07 & 0.11 & 0.06 & 0.02 & 0.10 & 0.12 & 0.13 & 0.12 & 0.21 & -0.08 & 0.79 \\
 & \texttt{quasi-inf-agg} & -0.70 & -0.04 & 0.07 & 0.13 & 0.16 & 0.16 & 0.22 & 0.28 & 0.24 & 0.30 & 0.26 & 0.33 & 0.33 & 0.34 & 0.34 & 0.46 & -0.09 & 0.75 \\
 & \texttt{quasi-inf-cons} & -0.66 & -0.02 & 0.08 & 0.14 & 0.19 & 0.17 & 0.21 & 0.28 & 0.27 & 0.28 & 0.20 & 0.34 & 0.33 & 0.33 & 0.30 & 0.44 & -0.09 & 0.73 \\
 & \textbf{Sum} & \textbf{-1.71} & \textbf{-0.31} & \textbf{0.00} & \textbf{0.12} & \textbf{0.44} & \textbf{0.43} & \textbf{0.49} & \textbf{0.67} & \textbf{0.73} & \textbf{0.75} & \textbf{0.58} & \textbf{0.95} & \textbf{0.98} & \textbf{0.99} & \textbf{0.91} & \textbf{1.40} & \textbf{-0.36} & \textbf{0.82} \\
\addlinespace[0.5ex]
\multirow{5}{*}{\texttt{Regime model 6} (unknown)} & \texttt{agg} & -0.29 & -0.10 & -0.09 & -0.05 & 0.04 & 0.05 & 0.04 & 0.05 & 0.08 & 0.09 & 0.06 & 0.12 & 0.15 & 0.15 & 0.12 & 0.22 & -0.12 & 0.83 \\
 & \texttt{cons} & 0.00 & -0.01 & 0.05 & 0.03 & 0.10 & 0.06 & 0.05 & 0.08 & 0.11 & 0.08 & 0.06 & 0.10 & 0.17 & 0.15 & 0.15 & 0.16 & -0.00 & 0.86 \\
 & \texttt{quasi-inf-agg} & -0.41 & 0.03 & 0.08 & 0.12 & 0.18 & 0.16 & 0.18 & 0.24 & 0.23 & 0.22 & 0.14 & 0.26 & 0.30 & 0.30 & 0.30 & 0.34 & -0.03 & 0.85 \\
 & \texttt{quasi-inf-cons} & -0.41 & 0.05 & 0.10 & 0.14 & 0.20 & 0.16 & 0.17 & 0.24 & 0.25 & 0.23 & 0.18 & 0.25 & 0.29 & 0.28 & 0.28 & 0.33 & -0.02 & 0.82 \\
 & \textbf{Sum} & \textbf{-1.11} & \textbf{-0.03} & \textbf{0.13} & \textbf{0.24} & \textbf{0.52} & \textbf{0.43} & \textbf{0.44} & \textbf{0.61} & \textbf{0.67} & \textbf{0.62} & \textbf{0.44} & \textbf{0.74} & \textbf{0.91} & \textbf{0.87} & \textbf{0.85} & \textbf{1.04} & \textbf{-0.17} & \textbf{0.87} \\
\addlinespace[0.5ex]
\multirow{5}{*}{\texttt{Regime model 7} (known)} & \texttt{agg} & -0.40 & -0.29 & -0.04 & 0.02 & 0.08 & 0.08 & 0.13 & 0.21 & 0.25 & 0.13 & -0.03 & -0.12 & 0.09 & 0.30 & 0.10 & 0.26 & -0.10 & 0.28 \\
 & \texttt{cons} & 0.14 & -0.00 & 0.03 & -0.09 & -0.02 & 0.04 & -0.04 & 0.02 & 0.05 & 0.11 & -0.18 & -0.15 & -0.12 & -0.16 & 0.06 & -0.17 & 0.07 & 0.40 \\
 & \texttt{quasi-inf-agg} & -0.39 & -0.07 & 0.04 & 0.27 & 0.25 & 0.22 & 0.28 & 0.41 & 0.51 & 0.27 & 0.24 & 0.35 & 0.28 & 1.08 & 0.43 & 0.88 & -0.09 & 0.53 \\
 & \texttt{quasi-inf-cons} & -0.19 & -0.12 & 0.05 & 0.04 & 0.04 & 0.10 & 0.04 & 0.03 & 0.17 & 0.06 & -0.07 & -0.10 & -0.03 & -0.05 & 0.05 & 0.13 & -0.03 & 0.21 \\
 & \textbf{Sum} & \textbf{-0.85} & \textbf{-0.48} & \textbf{0.09} & \textbf{0.24} & \textbf{0.35} & \textbf{0.44} & \textbf{0.42} & \textbf{0.67} & \textbf{0.97} & \textbf{0.58} & \textbf{-0.03} & \textbf{-0.01} & \textbf{0.21} & \textbf{1.17} & \textbf{0.63} & \textbf{1.10} & \textbf{-0.15} & \textbf{0.34} \\
\addlinespace[0.5ex]
\multirow{5}{*}{\texttt{Regime model 7} (unknown)} & \texttt{agg} & -0.26 & -0.07 & -0.09 & -0.04 & 0.01 & 0.10 & 0.06 & -0.21 & 0.07 & -0.12 & -0.31 & 0.12 & -0.01 & -0.13 & 0.04 & 0.05 & -0.08 & 0.02 \\
 & \texttt{cons} & -0.01 & 0.06 & 0.10 & 0.13 & 0.09 & 0.07 & 0.04 & -0.12 & 0.08 & -0.00 & -0.30 & 0.01 & 0.04 & -0.13 & 0.05 & -0.13 & 0.07 & 0.20 \\
 & \texttt{quasi-inf-agg} & -0.10 & 0.04 & 0.12 & 0.18 & 0.14 & 0.19 & 0.22 & 0.17 & 0.30 & 0.22 & 0.06 & 0.22 & 0.33 & 0.18 & 0.28 & 0.18 & 0.08 & 0.39 \\
 & \texttt{quasi-inf-cons} & -0.24 & 0.06 & 0.06 & 0.10 & 0.12 & 0.14 & 0.14 & 0.12 & 0.24 & 0.06 & -0.09 & 0.21 & 0.21 & 0.07 & 0.11 & 0.11 & 0.03 & 0.12 \\
 & \textbf{Sum} & \textbf{-0.62} & \textbf{0.08} & \textbf{0.18} & \textbf{0.36} & \textbf{0.37} & \textbf{0.50} & \textbf{0.46} & \textbf{-0.04} & \textbf{0.70} & \textbf{0.15} & \textbf{-0.64} & \textbf{0.55} & \textbf{0.57} & \textbf{-0.01} & \textbf{0.48} & \textbf{0.20} & \textbf{0.10} & \textbf{0.04} \\
\bottomrule
\end{tabular}%
}
\caption{Mean per-step wealth change by regime model, policy, and informed ratio $\psi$. The \textbf{Sum} row reports the sum across the four policies within each regime specification. The final three columns give OLS estimates from regressing mean per-step wealth change on the informed ratio $\psi$. All values are rounded to two decimals. Here \texttt{agg}, \texttt{cons}, \texttt{quasi-inf-agg}, and \texttt{quasi-inf-cons} denote, respectively, the Aggressive, Conservative, Quasi-informed aggressive, and Quasi-informed conservative strategies.}
\vspace{3mm}
\footnotesize
\begin{tabular}{@{}ll@{}}
\textbf{Regime mapping (table order):} & \\
Regime model 1 & \texttt{One change at a fixed time to a fixed regime} \\
Regime model 2 & \texttt{One change at a fixed time to a random regime} \\
Regime model 3 & \texttt{One change at a random time to a fixed regime} \\
Regime model 4 & \texttt{One change at a random time to a random regime} \\
Regime model 5 & \texttt{Fixed number of changes on a fixed grid to random regimes} \\
Regime model 6 & \texttt{Fixed number of changes on a random grid to random regimes} \\
Regime model 7 & \texttt{Exponential waiting times for random regimes} \\
\end{tabular}
\label{tab:mean_wealth_change_by_psi}
\end{table}
\FloatBarrier

\section{Conclusions}

The main result of this paper is that market-making profitability is, on average, increasing in market informedness. In the simulated environment, higher values of the informed ratio $\psi$ are generally associated with higher terminal wealth, with the strongest improvement occurring when moving away from the least informed market regimes. At the same time, the relationship is not strictly monotonic: at high levels of informedness it becomes flatter and occasionally locally non-monotonic. The non-stationary experiments show that this pattern is not an artifact of the stationary baseline. The positive association remains visible across a broad class of regime-switching specifications, although it becomes weaker and less regular in the most irregular environments.

At first sight, this conclusion appears to differ from the findings of \cite{BMB2025}, where a larger share of informed traders worsens market-maker performance. However, the difference is naturally explained by the underlying market structure. In \cite{BMB2025}, the market maker quotes around an exogenous reference price under simpler order-flow assumptions. In the present setting, by contrast, prices are formed endogenously through the interaction of market and limit orders, order flow is self-exciting, and trade sizes are heterogeneous. As a result, informed trading has two opposing effects: it increases adverse-selection risk, but it also improves price discovery by reducing the persistence of mispricing between transaction prices and fundamentals. The results suggest that, outside the least informed environments, the price-discovery channel can dominate the toxic-flow channel.

More broadly, the paper shows that in richer order-book environments, informed trading need not reduce profitability; instead, it may improve market makers' performance by limiting prolonged deviations from fundamental value.

\subsection{Limitations and future work}

The analysis remains subject to several limitations that also point to natural directions for future research. First, although the model captures endogenous price formation, self-exciting order flow, and strategic competition, it is not yet tightly calibrated to high-frequency order-book data. A closer empirical calibration of queue dynamics, latency, participant heterogeneity, and other market frictions would help assess the external validity of the results more precisely. Second, the informational environment remains stylized, relying on a binary informed--uninformed decomposition. Extending the model to richer information structures and allowing agents to infer market regimes from partially observed signals would bring the setup closer to actual trading environments. Finally, an important theoretical extension would be to characterize more sharply the threshold at which the price-discovery effect of informed trading ceases to dominate the toxic-flow effect. Such an analysis could connect the reinforcement-learning evidence developed here more tightly to the comparative statics of the analytical market-microstructure literature.

\section{Declarations}
\subsection{Funding}
The authors declare that no funds, grants, or other support were received during the preparation of this manuscript.
\subsection{Competing Interests}
The authors have no relevant financial or non-financial interests to disclose.
\subsection{Author Contributions}
Konrad Ochedzan contributed to the study conception and design, methodology, software implementation, numerical experiments, formal analysis, visualization, and writing of the original draft. Nino Antulov-Fantulin contributed to the study conception and design, methodology, and supervision. Both authors read and approved the final manuscript.
\subsection{Code Availability}
The code used to generate the simulation results and figures is available at \href{https://github.com/konradochedzan/Market-Making-with-Information}{Github}.
\subsection{Gen AI Use}
Generative Artificial Intelligence, namely ChatGPT from OpenAI, has been used as an editorial and programming aid for selected code and text revisions. The research design, mathematical formulation, implementation decisions, validation of outputs, and final responsibility for the code and manuscript content remain with Konrad Ochedzan. ChatGPT is not an author or contributor in the scholarly-authorship sense, and its use should be understood as assisted editing and development support under human review.
\begin{appendices}
\section{Proofs for finite-horizon stability}\label{app:stability_proofs}

This appendix proves the stability statements in Section~\ref{sec:hawkes_stability}. Throughout, $N([0,T])=\sum_{i=1}^{C}N_i([0,T])$, $C=2K$, and $\lambda_i(t)=G_i(x(t^-))\lambda_i^H(t)$ denotes the predictable effective intensity. The convention $I_{t^-}=S_{t^-}-\Pmicro_{t^-}$ is used at possible jump times of the book. Since the fundamental process is continuous in the model, $S_{t^-}=S_t$.

\subsection{Corrective and destabilizing events}

At any time $t$, define the predictable sets of corrective and destabilizing types by
\begin{equation}
  \mathrm{corr}(t) :=
  \begin{cases}
    \{1,\ldots,K\}, & \text{if } I_{t^-}\ge 0,\\
    \{K+1,\ldots,2K\}, & \text{if } I_{t^-}<0,
  \end{cases}
  \qquad
  \mathrm{dest}(t):=\{1,\ldots,2K\}\setminus\mathrm{corr}(t).
  \label{eq:corr_dest_sets}
\end{equation}
Thus, if the microprice is below the fundamental value, buys are corrective; if the microprice is above the fundamental value, sells are corrective. The convention at $I_{t^-}=0$ is immaterial for the bounds.

\begin{lemma}[Destabilizing gain is at most one]\label{lem:dest_gain}
Under Assumptions \textup{(S1)--(S6)}, $G_i(x(t^-))\le1$ for every $t\in[0,T]$ and every $i\in\mathrm{dest}(t)$. Economically, order flow that would push prices farther away from fundamentals is not amplified by the informedness mechanism; it remains at its baseline Hawkes level or is dampened.
\end{lemma}

\begin{proof}
Suppose first that $I_{t^-}>0$. Then sells are destabilizing and
\[
  G^b(x(t^-))
  =\psi e^{-\eta\delta^b_{t^-}-\zeta I_{t^-}}
  +(1-\psi)e^{-\eta\delta^b_{t^-}}.
\]
Both exponents are non-positive because $\eta,\zeta>0$, $\delta^b_{t^-}\ge0$, and $I_{t^-}>0$. Hence $G^b(x(t^-))\le1$. The case $I_{t^-}<0$ is symmetric. When $I_{t^-}=0$, the destabilizing side under the convention in \eqref{eq:corr_dest_sets} has gain $e^{-\eta\delta^b_{t^-}}\le1$.
\end{proof}

\begin{lemma}[Corrective intensity lower bound]\label{lem:corr_lower}
Let $\gamma$ be defined by \eqref{eq:gamma}. At every time $t$ such that the corrective side of the book is non-empty at $t^-$,
\begin{equation}
  \sum_{i\in\mathrm{corr}(t)}\lambda_i(t)
  \ge
  \gamma e^{\zeta\abs{I_{t^-}}}.
  \label{eq:corr_lower}
\end{equation}
Economically, this lower bound is the model's correction channel: whenever liquidity is available on the side that would move prices toward fundamentals, informed traders become more aggressive as the asset becomes more underpriced or overpriced.
\end{lemma}

\begin{proof}
Suppose $I_{t^-}>0$ and the ask side is non-empty. Then corrective types are buys. For every $i\le K$,
\begin{align*}
  \lambda_i(t)
  &=G^a(x(t^-))\lambda_i^H(t)\\
  &\ge \psi e^{-\eta\delta^a_{t^-}+\zeta I_{t^-}}\lambda_i^H(t)\\
  &\ge \psi e^{-\eta\dmax}\mu_i e^{\zeta I_{t^-}}
  \ge \gamma e^{\zeta I_{t^-}},
\end{align*}
where we used $\delta^a_{t^-}\le\dmax$ and $\lambda_i^H(t)\ge\mu_i$. Summing over the corrective types gives \eqref{eq:corr_lower}. The case $I_{t^-}<0$ is symmetric, with sells as the corrective types and $\abs{I_{t^-}}=-I_{t^-}$. At $I_{t^-}=0$, the same argument gives the lower bound $\gamma$.
\end{proof}

\subsection{Event-count bounds}

\begin{proof}[Proof of Proposition~\ref{prop:nonexplosion}]
Fix an inter-decision interval $[\tdec_n,\tdec_{n+1})$. By Assumptions \textup{(S1)--(S6)}, the total resting depth on each side at $\tdec_n$ is at most $\Dtot$, and no new limit orders are posted before $\tdec_{n+1}$. Each market buy consumes at least $\qmin$ units from the ask side, and each market sell consumes at least $\qmin$ units from the bid side. Once a side is depleted, the corresponding gain is zero, so no further events of that type can execute before the next decision time. Therefore,
\[
  N([\tdec_n,\tdec_{n+1}))
  \le \frac{\Dtot}{\qmin}+\frac{\Dtot}{\qmin}
  =\frac{2\Dtot}{\qmin}.
\]
Summing over the $\Nsteps$ inter-decision intervals yields $N([0,T])\le2\Dtot\Nsteps/\qmin=\Nmax$ almost surely. Taking expectations gives the final claim.
\end{proof}

The next auxiliary proposition is not needed for the pathwise non-explosion bound, but it records how the destabilizing part of the process is controlled by the Hawkes branching structure. Let $A$ be the integrated kernel matrix in \eqref{eq:integrated_kernel} and $\|A\|_1:=\max_{1\le j\le C}\sum_{i=1}^{C}A_{ij}$.

\begin{proposition}[Auxiliary moment bound]\label{prop:moment}
If, in addition to Assumptions \textup{(S1)--(S6)}, $\|A\|_1<1$, then
\begin{equation}
  \E[N([0,T])]
  \le
  \frac{\Nmax+\|\boldsymbol{\mu}\|_1T}{1-\|A\|_1},
  \label{eq:moment_bound}
\end{equation}
where $\|\boldsymbol{\mu}\|_1=\sum_{i=1}^{C}\mu_i$. If $\tdec_{n+1}-\tdec_n\ge\Dtmin>0$, then
\begin{equation}
  \E[N([0,T])]
  \le
  \frac{2\Dtot/(\qmin\Dtmin)+\|\boldsymbol{\mu}\|_1}{1-\|A\|_1}\,T.
  \label{eq:moment_linear}
\end{equation}
Economically, this bound allows the model to generate clustered high-frequency order flow without making total trading volume explode, which is essential for comparing market-maker profitability across informedness regimes.
\end{proposition}

\begin{proof}
Let $N^{\mathrm{corr}}$ and $N^{\mathrm{dest}}$ count corrective and destabilizing events according to \eqref{eq:corr_dest_sets}. Corrective events consume finite depth, so the same depth argument used above gives
\begin{equation}
  N^{\mathrm{corr}}([0,T])\le\Nmax
  \qquad\text{a.s.}
  \label{eq:corr_bound}
\end{equation}
By Lemma~\ref{lem:dest_gain}, destabilizing gains are bounded by one. Hence, using the compensator identity and non-negativity of the raw Hawkes intensities,
\begin{equation}
  \E[N^{\mathrm{dest}}([0,T])]
  \le
  \E\!\left[\int_0^T\sum_{i=1}^{C}\lambda_i^H(t)\,\dd t\right].
  \label{eq:dest_comp}
\end{equation}
For the raw Hawkes process,
\[
  \int_0^T\sum_{i=1}^{C}\lambda_i^H(t)\,\dd t
  =\|\boldsymbol{\mu}\|_1T
  +\sum_{i,j}\int_0^T\int_0^{t^-}\phi_{i\leftarrow j}(t-s)\,\dd N_j(s)\,\dd t.
\]
Fubini--Tonelli and the definition of $A_{ij}$ give
\[
  \int_0^T\int_0^{t^-}\phi_{i\leftarrow j}(t-s)\,\dd N_j(s)\,\dd t
  \le A_{ij}N_j([0,T]).
\]
Therefore,
\[
  \E\!\left[\int_0^T\sum_{i=1}^{C}\lambda_i^H(t)\,\dd t\right]
  \le
  \|\boldsymbol{\mu}\|_1T
  +\sum_{i,j}A_{ij}\E[N_j([0,T])]
  \le
  \|\boldsymbol{\mu}\|_1T+\|A\|_1\E[N([0,T])].
\]
Combining this inequality with \eqref{eq:corr_bound} and \eqref{eq:dest_comp} yields
\[
  \E[N([0,T])]
  \le
  \Nmax+\|\boldsymbol{\mu}\|_1T+\|A\|_1\E[N([0,T])].
\]
Rearranging gives \eqref{eq:moment_bound}. The linear bound \eqref{eq:moment_linear} follows from $\Nsteps\le T/\Dtmin$.
\end{proof}

\subsection{Exponential integrability and occupation time}

\begin{proof}[Proof of Proposition~\ref{prop:exp_int}]
Define the aggregate corrective intensity
\[
  \Lambda^{\mathrm{corr}}(t):=\sum_{i\in\mathrm{corr}(t)}\lambda_i(t).
\]
It is predictable, and by the compensator identity,
\[
  \E\!\left[\int_0^T\Lambda^{\mathrm{corr}}(t)\,\dd t\right]
  =\E[N^{\mathrm{corr}}([0,T])]
  \le\Nmax,
\]
where the final inequality is \eqref{eq:corr_bound}. On the set where the corrective side is non-empty, Lemma~\ref{lem:corr_lower} gives $\Lambda^{\mathrm{corr}}(t)\ge\gamma e^{\zeta\abs{I_{t^-}}}$. When the corrective side is empty, the corresponding gain is zero. Hence
\[
  \gamma\E\!\left[\int_0^T e^{\zeta\abs{I_{t^-}}}
  \ind_{\{\text{corrective side non-empty at }t^-\}}\,\dd t\right]
  \le
  \E\!\left[\int_0^T\Lambda^{\mathrm{corr}}(t)\,\dd t\right]
  \le\Nmax,
\]
which proves \eqref{eq:exp_int}.

For the full-interval bound, decompose
\begin{align*}
  \int_0^T e^{\zeta\abs{I_{t^-}}}\,\dd t
  &=
  \int_0^T e^{\zeta\abs{I_{t^-}}}
    \ind_{\{\text{corrective side non-empty at }t^-\}}\,\dd t\\
  &\quad+
  \int_0^T e^{\zeta\abs{I_{t^-}}}
    \ind_{\{\text{corrective side empty at }t^-\}}\,\dd t.
\end{align*}
The first term is controlled by \eqref{eq:exp_int}; the second is at most $e^{\zeta\sup_{t\le T}\abs{I_{t^-}}}T_{\mathrm{ecs}}$. Taking expectations proves \eqref{eq:exp_int_tecs}.
\end{proof}

\begin{proof}[Proof of Corollary~\ref{cor:occupation}]
For every $M>0$,
\[
  \ind_{\{\abs{I_{t^-}}>M\}}
  \le e^{-\zeta M}e^{\zeta\abs{I_{t^-}}}.
\]
Multiplying by the indicator that the corrective side is non-empty, integrating, and applying \eqref{eq:exp_int} gives \eqref{eq:occupation}. For the full-interval bound, split the occupation time according to whether the corrective side is empty:
\begin{align*}
  \E\!\left[\int_0^T\ind_{\{\abs{I_{t^-}}>M\}}\,\dd t\right]
  &\le
  \E\!\left[\int_0^T\ind_{\{\abs{I_{t^-}}>M\}}
  \ind_{\{\text{corrective side non-empty at }t^-\}}\,\dd t\right]
  +\E[T_{\mathrm{ecs}}]\\
  &\le \frac{\Nmax}{\gamma}e^{-\zeta M}+\E[T_{\mathrm{ecs}}].
\end{align*}
\end{proof}

\subsection{Pathwise mispricing control}

\begin{lemma}[Deterministic microprice bound]\label{lem:Pmicro_bound}
Under Assumptions \textup{(S1)--(S6)},
\begin{equation}
  \sup_{t\in[0,T]}\abs{\Pmicro_t-\Pmicro_0}
  \le
  2\dmax(\Nmax+\Nsteps)
  =\Delta_P.
  \label{eq:Pmicro_bound}
\end{equation}
Economically, the learned market makers operate in a discrete LOB with controlled quote distances, so endogenous price formation remains bounded even when liquidity demand is clustered.
\end{lemma}

\begin{proof}
The microprice changes only at trade times and decision times. By Proposition~\ref{prop:nonexplosion}, there are at most $\Nmax$ trade times, and by construction there are at most $\Nsteps$ decision times. Between such times, the book state is constant and so is $\Pmicro_t$.

Whenever both sides are non-empty,
\[
  \Pmicro_t=\frac{A_tQ_t^b+B_tQ_t^a}{Q_t^a+Q_t^b}\in[B_t,A_t].
\]
Assumptions \textup{(S1)--(S6)} give $A_t-\Pmicro_t\le\dmax$ and $\Pmicro_t-B_t\le\dmax$, hence $A_t-B_t\le2\dmax$. At a trade or decision time, the old and new microprices are therefore separated by at most $2\dmax$. If a side becomes empty, the fallback value $\Plast_t$ is an executed price at a level within the same admissible range, so the same one-step bound applies. Summing over at most $\Nmax+\Nsteps$ changes proves \eqref{eq:Pmicro_bound}.
\end{proof}

\begin{proof}[Proof of Proposition~\ref{prop:pathwise}]
For every $t\in[0,T]$,
\begin{align*}
  \abs{I_t}
  &=\abs{S_t-\Pmicro_t}\\
  &=\abs{(S_t-S_0)-(\Pmicro_t-\Pmicro_0)+(S_0-\Pmicro_0)}\\
  &\le \abs{S_t-S_0}+\abs{\Pmicro_t-\Pmicro_0}+\abs{I_0}.
\end{align*}
Taking suprema and using Lemma~\ref{lem:Pmicro_bound} gives
\[
  \sup_{t\le T}\abs{I_t}
  \le
  \sup_{t\le T}\abs{S_t-S_0}+\Delta_P+\abs{I_0}
  =
  \sup_{t\le T}\abs{S_t-S_0}+K_0.
\]
Thus, for $K>K_0$,
\[
  \Prob\!\left(\sup_{t\le T}\abs{I_t}\le K\right)
  \ge
  \Prob\!\left(\sup_{t\le T}\abs{S_t-S_0}\le K-K_0\right).
\]
Since $S_t-S_0=\sigma W_t$, the reflection principle and the standard Gaussian tail bound imply, for $x>0$,
\[
  \Prob\!\left(\sup_{t\le T}\abs{W_t}>x\right)
  \le 2\Prob(\abs{W_T}>x)
  \le 2\exp\!\left(-\frac{x^2}{2T}\right).
\]
Applying this with $x=(K-K_0)/\sigma$ yields \eqref{eq:pathwise_tail}. Sending $K\to\infty$ gives $\Prob(\sup_{t\le T}\abs{I_t}<\infty)=1$.
\end{proof}

The moment claim in the main-text remark follows from the standard tail-integration identity. For example, applying \eqref{eq:pathwise_tail} to $X=\sup_{t\le T}\abs{I_t}$ gives
\[
  \E[X^p]
  =p\int_0^{\infty}K^{p-1}\Prob(X>K)\,\dd K
  \le K_0^p+p\int_{K_0}^{\infty}K^{p-1}2\exp\!\left(-\frac{(K-K_0)^2}{2\sigma^2T}\right)\dd K,
\]
which is finite. The same Gaussian tail yields the stated sub-Gaussian-square integrability for every $\lambda<1/(2\sigma^2T)$.

For reference, the main constants appearing in the stability bounds are
\begin{align*}
  \Nmax &= \frac{2\Dtot\Nsteps}{\qmin}
  \le \frac{2\Dtot T}{\qmin\Dtmin},\\
  \gamma &= \psi e^{-\eta\dmax}\mu_{\min},\\
  \Delta_P &= 2\dmax(\Nmax+\Nsteps),\\
  K_0 &= \abs{I_0}+\Delta_P.
\end{align*}

\section{Cancellation mechanisms}
\label{app:cancel}

Each market maker stores its outstanding quotes by price-level and maintains an \emph{age counter} per quote. The age is incremented once per decision step (i.e.\ on the discrete market-maker grid), so the time-to-live (TTL) is expressed in \emph{steps}, not in continuous time. Deterministic cancellations occur at three points within each decision step:
\begin{enumerate}
    \item \textbf{Time-to-live (TTL) at the beginning of the step.}
    At the beginning of each market-maker decision step, all existing quotes have their age incremented and any quote with age exceeding a fixed threshold is removed. To prevent pathological behavior where agents are forced to constantly repost their best quotes (thereby losing queue priority), the environment protects the closest quotes to the current mid on each side (closest in absolute tick distance). Thus, an old quote is removed only if it is both expired and not among the protected top-$k$ levels.

    \item \textbf{Enforcing per-side outstanding-order caps before posting.}
    Immediately before posting a new bid (respectively ask), the environment enforces two per-side caps:
    (i) a cap on the number of distinct active price levels on that side (denoted as $S_{\text{own}}>0$) and
    (ii) a cap on the total outstanding quantity on that side.
    The environment never rejects the agent's posting action; instead, it makes room by canceling existing levels on that side until both caps are satisfied for the post-trade state including the new order.
    The default cancellation rule implemented in code cancels the level farthest from the current mid on that side (least competitive), with ties broken by age (oldest first). Formally, letting $p$ denote a candidate price level and $\mathrm{age}(p)$ its step-age, the environment ranks candidates by
    \[
        d_{\mathrm{bid}}(p) = M_t - p, \qquad d_{\mathrm{ask}}(p) = p - M_t,
    \]
    and cancels the level maximizing the lexicographic key $(d_{\text{side}}(p), \mathrm{age}(p))$. Importantly, the cap logic is applied only if the incoming order size is strictly positive, so an action corresponding to size $0$ (``no new order'') does not trigger cancellations and does not interact with the caps.

    \item \textbf{Bounding quotes after continuous-time taker trading.}
    After the market evolves in continuous time over the interval $(t^{\mathrm{dec}}_n,t^{\mathrm{dec}}_{n+1}]$ due to market-taker arrivals, the mid-price is recomputed and any remaining quotes that drift beyond a fixed admissible distance from the updated mid are removed to keep the book bounded. Concretely, define $L$ as maximum observable horizon. Then all bids satisfying $M_t - p > L\cdot \Delta p$ and all asks satisfying $p - M_t > L\cdot \Delta p$ are deleted from both the order book and the market maker's internal records.
\end{enumerate}
Together, these rules remove cancellation degeneracy from the action space, reduce the frequency of extreme states caused by unbounded stale quotes, and make the effect of a placement action substantially more predictable from the agent's perspective.

\section{Observation space decomposition}
\label{app:observation}
The observation space can be divided into three parts: the actors' observable state, the global state used by the critic and an agent indicator. We start of by describing the agents' observation space. 
\paragraph{Agents' observation space}
Each agents' observation space can be further qualitatively divided into three sections - (i) scalars, (ii) LOB embedding, and (iii) history. This division matters greatly in the architecture. Below we describe the construction of each of these parts.

\noindent\textbf{(i) Actor scalar features.}
The scalar block contains the agent's own inventory and cash, time elapsed measure, market price, last transaction price together with a (possibly absent) private signal:
\[
u_t^{(i)} = \Big(
q_t^{(i)},\;
c_t^{(i)},\;
M_t,\;
P_t^{\mathrm{last}},\;
\tau_t,\;
\widetilde{S}_t^{(i)},\;
\mathbf{1}\{\text{agent $i$ informed}\}
\Big)\in\mathbb{R}^{7}.
\]
Here $\mathbf{1}\{\cdot\}$ is the signal-availability indicator and $\tau_t\in[0,1]$ is normalized time elapsed in the episode. Optionally, when regime information is exposed to the actor, this block is augmented by
\[
r_t = \big(k_t,\psi_t^{\mathrm{obs}},\sigma_t^{\mathrm{obs}}\big)\in\mathbb{R}^3,
\]
where $k_t \in \{ 0,\ 1\}$ denotes whether the regime is known to agents and $\psi_t^{\mathrm{obs}}, \ \sigma_t^{\mathrm{obs}}$ are respectively observed market volatility and informedness regime. Otherwise these variables are omitted from the actor input.

\medskip
\noindent\textbf{(ii) Public LOB features and liquidity decomposition.}
Let $L$ denote the maximum observable distance from mid in ticks and let tick size be denoted as $\Delta p$. Everything farther than $L$ ticks from the mid-price is canceled (see Appendix \ref{app:cancel}), and agents observe the LOB only within $L$ ticks of the mid-price after canonicalization. Using the level grid defined by:
\[
p^{b}_{t,\ell} = \mathrm{canon}\!\left(R_t - \ell\,\Delta p\right),\qquad
p^{a}_{t,\ell} = \mathrm{canon}\!\left(R_t + \ell\,\Delta p\right),
\qquad \ell=1,\dots,L,
\]
where the reference price $R_t=M_t$ if both sides exist and $R_t=P^{\mathrm{last}}_t$ otherwise. Recall the FIFO queue structure defined in Equation \ref{def:queue}. Thanks to this construction, rather than collapsing each level to a single depth number, the observation exposes a truncated queue representation that preserves time-priority structure up to a fixed depth. Fix an integer $S_{\mathrm{own}}\ge 1$ stating how many different quotes are allowed per agent at one price level and define $S_{\mathrm{oth}}=S_{\mathrm{own}}+1$. Importantly, this is not a cancellation mechanism but a posting constraint. If the agent already occupies the maximum number of own queue slots at a given price level, an additional order at that level is ignored. In regular market conditions this constraint has negligible impact on the simulated dynamics, while preventing pathological over-posting during training and preserving queue-priority effects. For a fixed agent $i$, we define four arrays
\[
V^{\mathrm{own},b}_{t}(i)\in\mathbb{R}^{L\times S_{\mathrm{own}}},\quad
V^{\mathrm{own},a}_{t}(i)\in\mathbb{R}^{L\times S_{\mathrm{own}}},
\qquad
V^{\mathrm{oth},b}_{t}(i)\in\mathbb{R}^{L\times S_{\mathrm{oth}}},\quad
V^{\mathrm{oth},a}_{t}(i)\in\mathbb{R}^{L\times S_{\mathrm{oth}}}.
\]

in the following fashion. For any bid-side price level $p\in\mathcal{B}_t$, recall from Definition~\ref{def:queue} that
\[
\mathcal{Q}^b_t(p)
=
\big((j^b_{1,p},q^b_{1,p}),\dots,(j^b_{n_p,p},q^b_{n_p,p})\big)_t.
\]
For a fixed agent $i$, define the set of FIFO positions occupied by that agent at price $p$ by
\[
\mathcal{I}^b_t(p,i)
\defeq
\big\{m\in\{1,\dots,n_p\}: j^b_{m,p}=i\big\}.
\]
Let
\[
K^b_t(p,i)\defeq |\mathcal{I}^b_t(p,i)|
\]
be the number of stored own bid entries of agent $i$ at price $p$. Notice that then $K^b_t(p,i)\leq \ S_{\mathrm{own}}$. If $K^b_t(p,i)\ge 1$, let
\[
\pi^b_{t,1}(p,i)<\pi^b_{t,2}(p,i)<\dots<\pi^b_{t,K^b_t(p,i)}(p,i)
\]
denote the elements of $\mathcal{I}^b_t(p,i)$ in increasing FIFO order.

Now fix $\ell\in\{1,\dots,L\}$ and set $p=p^b_{t,\ell}$, so the $\ell$-th price level available on the bid side. The own-liquidity tensor is defined by
\[
V^{\mathrm{own},b}_{t,\ell,s}(i)
=
\begin{cases}
q^b_{\pi^b_{t,s}(p,i),\,p}, & 1\le s\le K^b_t(p,i),\\
0, & K^b_t(p,i)< s\le S_{\mathrm{own}}.
\end{cases}
\]
Thus, $V^{\mathrm{own},b}_{t,\ell,s}(i)$ stores the sizes of the $S_{\mathrm{own}}$ own bid orders of agent $i$ at level $\ell$, preserving FIFO priority.

The foreign-liquidity tensor
\[
V^{\mathrm{oth},b}_{t,\ell,s}(i),\qquad s=1,\dots,S_{\mathrm{oth}},\qquad S_{\mathrm{oth}}=S_{\mathrm{own}}+1,
\]
uses the same queue positions and has fixed slot semantics:
\begin{align*}
V^{\mathrm{oth},b}_{t,\ell,1}(i)
&=
\begin{cases}
\displaystyle\sum_{\substack{m<\pi^b_{t,1}(p,i)\\ j^b_{m,p}\neq i}} q^b_{m,p}, & K^b_t(p,i)\ge 1,\\[1ex]
0, & K^b_t(p,i)=0,
\end{cases}\\[1ex]
V^{\mathrm{oth},b}_{t,\ell,s}(i)
&=
\begin{cases}
\displaystyle\sum_{\substack{\pi^b_{t,s-1}(p,i)<m<\pi^b_{t,s}(p,i)\\ j^b_{m,p}\neq i}} q^b_{m,p},
& 2\le s\le K^b_t(p,i),\\[1ex]
0, & K^b_t(p,i)< s\le S_{\mathrm{own}},
\end{cases}\\[1ex]
V^{\mathrm{oth},b}_{t,\ell,S_{\mathrm{oth}}}(i)
&=
\begin{cases}
\displaystyle\sum_{\substack{m>\pi^b_{t,K^b_t(p,i)}(p,i)\\ j^b_{m,p}\neq i}} q^b_{m,p},
& K^b_t(p,i)\ge 1,\\[1ex]
\displaystyle\sum_{\substack{m=1,\dots,n_p\\ j^b_{m,p}\neq i}} q^b_{m,p},
& K^b_t(p,i)=0.
\end{cases}
\end{align*}
Hence the competing liquidity is decomposed into volume ahead of the first own order, volume between consecutive own orders, and volume behind the last stored own order. If agent $i$ has no own order at that price level, all competing liquidity is placed in the final slot $S_{\mathrm{oth}}$.

The ask-side tensors are defined analogously. For any ask-side price level $p\in\mathcal{A}_t$, write
\[
\mathcal{Q}^a_t(p)
=
\big((j^a_{1,p},q^a_{1,p}),\dots,(j^a_{n_p,p},q^a_{n_p,p})\big)_t,
\]
define $\mathcal{I}^a_t(p,i)$, $K^a_t(p,i)$ and the ordered positions $\pi^a_{t,s}(p,i)$ in the same way, and then instantiate the definitions above at $p=p^a_{t,\ell}$ (so $\ell$-th level on the ask side) to obtain $V^{\mathrm{own},a}_{t,\ell,s}(i)$ and $V^{\mathrm{oth},a}_{t,\ell,s}(i)$.

\medskip
\noindent This representation preserves the economically relevant time-priority information at each level: the agent observes not only how much competing liquidity exists at a given price, but also whether it is ahead of, between, or behind its own queue entries. At the same time, the resulting features are fixed-shape arrays, hence directly compatible with feed-forward policies.

\medskip
\noindent\textbf{(iii) History embedding.}
\label{def:history_emb}
The history tensor is defined as:
\[
H_t^{(i)}=\big[\xi_{t-H}^{(i)},\dots,\xi_{t-1}^{(i)}\big]^{\top}\in\mathbb{R}^{H\times 6},
\]
where
\[
\xi_{t}^{(i)} =
\Big(
B_t,\;
A_t,\;
P^{\mathrm{last}}_t,\;
\widetilde{S}_t^{(i)},\;
q_t^{(i)},\;
c_t^{(i)}
\Big)\in\mathbb{R}^{6}.
\]
The buffer is padded at episode start to provide a constant shape. 

\medskip
\noindent\textbf{Centralized critic state (CTDE).}
\label{def:global_state}
For centralized training, the critic receives a privileged state vector
\[
s_t \in \mathcal{S}.
\]
This vector is not part of the decentralized actor observation $o_t^{(i)}$. In addition, the environment returns a one-hot identifier
\[
e^{(i)}\in\{0,1\}^{N},
\]
indicating for which agent the value function is being evaluated.

In the implementation, $s_t$ is a single flat vector built from a global scalar block, all per-agent scalar blocks, and the market-wide queue tensors. Define
\[
g_t
=
\Big(
M_t,\;
P_t^{\mathrm{last}},\;
S_t,\;
\tau_t,\;
\psi_t^{\mathrm{obs}},\;
\sigma_t^{\mathrm{obs}}
\Big)\in\mathbb{R}^{6},
\]
where $\tau_t$ is normalized time elapsed, and where $\psi_t^{\mathrm{obs}}=\sigma_t^{\mathrm{obs}}=0$ whenever the current regime is hidden from the learning pipeline as previously.

For each agent $j\in\{1,\dots,N\}$, define the per-agent scalar block
\[
a_t^{(j)}
=
\Big(
q_t^{(j)},\;
c_t^{(j)},\;
\alpha_j,\;
\phi_j,\;
\widetilde{S}_t^{(j)},\;
\mathbf{1}\{\text{agent $j$ informed}\}
\Big)\in\mathbb{R}^{6},
\]

where $\alpha_j,\;\phi_j,\;$ are the policy specific inventory penalties that play a role in the reward structure in Section \ref{sec:rewards}. Using the same queue encoding as in the actor observation, let
\[
V_t^{\mathrm{own},b}(j)\in\mathbb{R}^{L\times S_{\mathrm{own}}},
\quad
V_t^{\mathrm{oth},b}(j)\in\mathbb{R}^{L\times S_{\mathrm{oth}}},
\quad
V_t^{\mathrm{own},a}(j)\in\mathbb{R}^{L\times S_{\mathrm{own}}},
\quad
V_t^{\mathrm{oth},a}(j)\in\mathbb{R}^{L\times S_{\mathrm{oth}}},
\]
denote the bid/ask own- and foreign-liquidity tensors from agent $j$'s perspective, where $S_{\mathrm{oth}}=S_{\mathrm{own}}+1$.

The centralized critic state is then the concatenation
\[
\begin{aligned}
s_t
=
\Big(
&g_t,\;
a_t^{(1)},\dots,a_t^{(N)},\\
&\operatorname{vec}\!\big(V_t^{\mathrm{own},b}(1)\big),\dots,\operatorname{vec}\!\big(V_t^{\mathrm{own},b}(N)\big),\\
&\operatorname{vec}\!\big(V_t^{\mathrm{oth},b}(1)\big),\dots,\operatorname{vec}\!\big(V_t^{\mathrm{oth},b}(N)\big),\\
&\operatorname{vec}\!\big(V_t^{\mathrm{own},a}(1)\big),\dots,\operatorname{vec}\!\big(V_t^{\mathrm{own},a}(N)\big),\\
&\operatorname{vec}\!\big(V_t^{\mathrm{oth},a}(1)\big),\dots,\operatorname{vec}\!\big(V_t^{\mathrm{oth},a}(N)\big)
\Big),
\end{aligned}
\]
where the ordering is agent-major, then level-major, then slot-major within each tensor.

Hence
\[
\dim(s_t)
=
6 + 6N + N\Big(LS_{\mathrm{own}} + LS_{\mathrm{oth}} + LS_{\mathrm{own}} + LS_{\mathrm{oth}}\Big)
=
6 + 6N + 2NL(2S_{\mathrm{own}}+1).
\]

\section{Agents' architecture}
\label{app:agents_net}
\paragraph{(i) Scalar encoder.}
Let $u_t^{(i)}\in\mathbb{R}^{d_{\mathrm{sc}}}$ denote the concatenation of the scalar keys used by the actor in code (see Appendix \ref{app:observation}).
These scalars are embedded by a one-layer MLP with $\tanh$ activation:
\[
h^{(i)}_{t,\mathrm{sc}}
=
\tanh\!\big(W_{\mathrm{sc}}\bar u_t^{(i)}+b_{\mathrm{sc}}\big)
\in\mathbb{R}^{d_{\text{sc\_out}}}.
\]
\paragraph{(ii) Liquidity cross attention.}
Recall that each agent observes two depth tensors describing its own outstanding orders and the liquidity posted by others:
\[
V^{\mathrm{own},b}_{t}(i)\in\mathbb{R}^{L\times S_{\mathrm{own}}},\quad
V^{\mathrm{own},a}_{t}(i)\in\mathbb{R}^{L\times S_{\mathrm{own}}},
\qquad
V^{\mathrm{oth},b}_{t}(i)\in\mathbb{R}^{L\times S_{\mathrm{oth}}},\quad
V^{\mathrm{oth},a}_{t}(i)\in\mathbb{R}^{L\times S_{\mathrm{oth}}},
\]
where $L$ is the number of tracked price levels (in ticks away from the reference) and $S_{\mathrm{own}},S_{\mathrm{oth}}$ are the per-level feature widths in the environment (see Appendix \ref{app:observation}).

For a fixed agent $i$ and level $\ell\in\{1,\dots,L\}$ we define the \emph{own} and \emph{other} level tokens by concatenation:
\[
\mathbf{u}^{(i)}_{t,\ell}
\; \defeq \;
\big[\,
V^{\mathrm{own},b}_{t}(i)[\ell,:],\;
V^{\mathrm{own},a}_{t}(i)[\ell,:]
\,\big]
\in\mathbb{R}^{2S_{\mathrm{own}}},
\]
\[\mathbf{v}^{(i)}_{t,\ell}
\; \defeq \;
\big[\,
V^{\mathrm{oth},b}_{t}(i)[\ell,:],\;
V^{\mathrm{oth},a}_{t}(i)[\ell,:]
\,\big]
\in\mathbb{R}^{2S_{\mathrm{oth}}}.
\]
Therefore, these vectors represent FIFO order liquidity at the $\ell$-th level from the mid-price on both sides, for both own and foreign liquidity. Although the bid and ask prices differ, the level index is defined relative to the mid-price in both cases. Stacking the level tokens yields two matrices
\[
U_t^{(i)}\in\mathbb{R}^{L\times 2S_{\mathrm{own}}},
\quad
V_t^{(i)}\in\mathbb{R}^{L\times 2S_{\mathrm{oth}}},
\qquad
U_t^{(i)}[\ell,:]=\mathbf{u}^{(i)}_{t,\ell},\ \ V_t^{(i)}[\ell,:]=\mathbf{v}^{(i)}_{t,\ell}.
\]

To combine own exposure with the surrounding liquidity landscape, the actor applies a single cross-attention layer in which each level of the agent's own token sequence queries all levels of the other-liquidity sequence. In particular, with key dimension $d_k$ and output dimension $d_{\mathrm{att}}$, the implementation computes
\[
Q = U_t^{(i)} W_q \in\mathbb{R}^{L\times d_k},\quad
K = V_t^{(i)} W_k \in\mathbb{R}^{L\times d_k},\quad
\mathcal{V} = V_t^{(i)} W_v \in\mathbb{R}^{L\times d_{\mathrm{att}}},
\]
followed by
\[
A = \mathrm{softmax}\!\left(\frac{QK^{\top}}{\sqrt{d_k}}\right)\in\mathbb{R}^{L\times L},
\qquad
C_t^{(i)} = A\,\mathcal{V}\in\mathbb{R}^{L\times d_{\mathrm{att}}}.
\]
Thus, for each level $\ell$, the row $C_{t,\ell}^{(i)}$ is a convex combination of transformed other-liquidity vectors across all levels, with weights determined by the compatibility between the agent's own queue configuration at $\ell$ and the competitive liquidity profile at every level.

\paragraph{History encoder.}
Using a history tensor $H_t^{(i)}\in\mathbb{R}^{H\times F}$ (as introduced in Appendix \ref{app:observation}), the actor embeds it with a single-layer GRU (\cite{cho2014gru}) of hidden size $d_h$:
\[
(h^{(i)}_{t,1},\dots,h^{(i)}_{t,H}) = \mathrm{GRU}\!\big(H_t^{(i)}\big),
\qquad
h^{(i)}_{t,\mathrm{hist}} \;\defeq\; h^{(i)}_{t,H}\in\mathbb{R}^{d_h}.
\]

\paragraph{Policy head.}
Let $\mathrm{vec}(\cdot)$ denote row-wise vectorization. The actor feature vector is
\[
z_t^{(i)} \;=\;
\Big[
h^{(i)}_{t,\mathrm{sc}},\;
h^{(i)}_{t,\mathrm{hist}},\;
\mathrm{vec}\!\big(C_t^{(i)}\big)
\Big]
\in \mathbb{R}^{d_{\text{sc\_out}} + d_h + d_{\text{att}}L}.
\]
A single policy head is a two-layer MLP with $\tanh$ nonlinearity and hidden width $d_{\text{head}}$:
\[
y_t^{(i)} = W_2\,\tanh(W_1 z_t^{(i)} + b_1) + b_2 \in \mathbb{R}^{d_\pi}.
\]

\paragraph{Mixture-of-experts (MoE) actor head.}
Policy heads are merged into a mixture-of-experts head. The only difference between agent types is the \emph{gating regime}: agents with access to a fundamental proxy (i.e.\ $\mathbf{1}\{\text{agent $i$ informed}\}=1$) use an inventory$\times$mispricing gate, whereas agents without that information use an inventory-only gate. 

Let $z_t^{(i)}$ be the actor feature vector defined above and let $d_\pi$ denote the total number of logits required by RLlib for the discrete action distribution. For a fixed number of experts $E$, each expert produces logits
\[
y_t^{(i,e)} = f_{\pi}^{(e)}(z_t^{(i)}) \in \mathbb{R}^{d_\pi},
\qquad e\in\{1,\dots,E\},
\]
where $f_{\pi}^{(e)}$ is a two-layer MLP with $\tanh$ activations and hidden width $256$. The final logits are a convex combination
\[
y_t^{(i)} = \sum_{e=1}^{E} w_t^{(i,e)}\,y_t^{(i,e)},
\qquad
w_t^{(i,e)}\ge 0,\quad \sum_{e=1}^{E} w_t^{(i,e)}=1.
\]
The gating weights $w_t^{(i,e)}$ are deterministic functions of low-dimensional indicators derived from the observation introduced in the next paragraph.

\medskip
\noindent\emph{Inventory gate (all agents).}
The inventory gate produces weights for negative vs.\ positive inventory regimes as a hard split,
\[
w_{\mathrm{pos}}^{(i)}=\mathbf{1}\{q_t^{(i)}>0\},\qquad
w_{\mathrm{neg}}^{(i)}=1-w_{\mathrm{pos}}^{(i)}.
\]

\medskip
\noindent\emph{Inventory-only MoE (uninformed agents).}
For agents with $\mathbf{1}\{\text{agent $i$ informed}\}=0$ the MoE uses $E=2$ experts corresponding to the two inventory regimes (neg/pos):
\[
y_t^{(i)}
=
w_{\mathrm{neg}}^{(i)}\,y_t^{(i,\mathrm{neg})}
+
w_{\mathrm{pos}}^{(i)}\,y_t^{(i,\mathrm{pos})}.
\]

\medskip
\noindent\emph{Inventory $\times$ mispricing MoE (agents with fundamental proxy).}
For agents with $\mathbf{1}\{\text{agent $i$ informed}\}=1$ an additional mispricing gate is applied. Using the actor-side proxy $\widetilde{S}_t^{(i)}$ and mid-price $M_t$, define the signed mispricing indicator
\[
\Delta_t^{(i)} := M_t - \widetilde{S}_t^{(i)}.
\]
The under/over weights $(w_{\mathrm{under}}^{(i)},w_{\mathrm{over}}^{(i)})$ are computed from $\Delta_t^{(i)}$ using the same hard gating form. The MoE then uses $E=4$ experts corresponding to the Cartesian product of regimes
\[
(\mathrm{neg},\mathrm{pos})\times(\mathrm{under},\mathrm{over}),
\]
with mixture weights given by products:
\[
\begin{aligned}
w_t^{(i,\mathrm{neg,under})}&=w_{\mathrm{neg}}^{(i)}w_{\mathrm{under}}^{(i)},&
w_t^{(i,\mathrm{neg,over})}&=w_{\mathrm{neg}}^{(i)}w_{\mathrm{over}}^{(i)},\\
w_t^{(i,\mathrm{pos,under})}&=w_{\mathrm{pos}}^{(i)}w_{\mathrm{under}}^{(i)},&
w_t^{(i,\mathrm{pos,over})}&=w_{\mathrm{pos}}^{(i)}w_{\mathrm{over}}^{(i)}.
\end{aligned}
\]
Finally,
\[
y_t^{(i)}
=
\sum_{r\in\{(\mathrm{neg,under}),(\mathrm{neg,over}),(\mathrm{pos,under}),(\mathrm{pos,over})\}}
w_t^{(i,r)}\,y_t^{(i,r)}.
\]
This design hard-codes economically meaningful regime structure (inventory sign and, when available, mispricing sign) while preserving a fully differentiable mapping from observations to action logits.

\section{Centralized critic's  architecture}
\label{app:critics_net}

\paragraph{Critic scalar encoder.}
The critic first builds a scalar feature vector by concatenating:
(i) the global scalars $g_t$,
(ii) all per-agent scalars $a_t^{(j)}$ augmented by an ``is-current-agent'' indicator extracted from $\mathbf{e}^{(i)}$,
and (iii) the history embedding (from the same GRU that is used by the actor) $h_{t,\mathrm{hist}}^{(i)}$ (see Appendix \ref{app:observation}).

Concretely, define the augmented per-agent tokens
\[
\tilde a_t^{(j\mid i)}=\big[a_t^{(j)},\ \mathbf{1}\{j=i\}\big]\in\mathbb{R}^{7},
\qquad
\tilde a_t^{(\cdot\mid i)}=\big[\tilde a_t^{(1\mid i)},\dots,\tilde a_t^{(N\mid i)}\big]\in\mathbb{R}^{7N},
\]
and set
\[
x_{t,\mathrm{sc}}^{(i)}=
\big[g_t,\ \tilde a_t^{(\cdot\mid i)},\ h_{t,\mathrm{hist}}^{(i)}\big].
\]
The scalar encoder is a single-layer MLP with $\tanh$ activation:
\[
h_{t,\mathrm{sc}}^{(i)}=\tanh\!\big(W_{\mathrm{sc}}x_{t,\mathrm{sc}}^{(i)}+b_{\mathrm{sc}}\big)\in\mathbb{R}^{d_{\text{sc\_out}}^{\text{crit}}}.
\]

\paragraph{Cross-attention over the market-wide queue state.}
Using the same per-level queue vectors as in the actors subsections, define for each agent $j$ and level $\ell$:
\[
\mathbf{w}_{t,\ell}^{(j)}
\;\defeq\;
\big[\,
\mathbf{u}_{t,\ell}^{(j)},\;
\mathbf{v}_{t,\ell}^{(j)}
\,\big]
\in\mathbb{R}^{F_q},
\qquad
F_q = 2S_{\mathrm{own}} + 2S_{\mathrm{oth}} = 4S_{\mathrm{own}} + 2.
\]
Stacking these tokens over all agents and all tracked levels yields
\[
W_t\in\mathbb{R}^{N\times L\times F_q},
\qquad
W_t[j,\ell,:]=\mathbf{w}_{t,\ell}^{(j)}.
\]
Using the one-hot identifier $\mathbf{e}^{(i)}$, define the current-agent query sequence
\[
U_t^{(i),\mathrm{critic}}\in\mathbb{R}^{L\times F_q},
\qquad
U_t^{(i),\mathrm{critic}}[\ell,:]
=
\sum_{j=1}^{N} e_j^{(i)}\,W_t[j,\ell,:]
=
W_t[i,\ell,:],
\]
and the market-wide context sequence
\[
V_t^{(i),\mathrm{critic}}\in\mathbb{R}^{(NL)\times F_q},
\qquad
V_t^{(i),\mathrm{critic}}[(j-1)L+\ell,:]
=
W_t[j,\ell,:].
\]
Then, in the same manner as in the actor part, the critic cross-attention block is written as
\begin{align*}
Q^{\mathrm{critic}}
=
U_t^{(i),\mathrm{critic}} W_q^{\mathrm{critic}}
\in\mathbb{R}^{L\times d_k^{\mathrm{critic}}},
\\
K^{\mathrm{critic}}
=
V_t^{(i),\mathrm{critic}} W_k^{\mathrm{critic}}
\in\mathbb{R}^{(NL)\times d_k^{\mathrm{critic}}},
\\
\mathcal{V}^{\mathrm{critic}}
=
V_t^{(i),\mathrm{critic}} W_v^{\mathrm{critic}}
\in\mathbb{R}^{(NL)\times d_{\mathrm{att}}^{\mathrm{critic}}},
\end{align*}
followed by
\[
A_t^{(i),\mathrm{critic}}
=
\mathrm{softmax}\!\left(
\frac{Q^{\mathrm{critic}}(K^{\mathrm{critic}})^{\top}}
{\sqrt{d_k^{\mathrm{critic}}}}
\right)
\in\mathbb{R}^{L\times (NL)},
\qquad
C_t^{(i),\mathrm{critic}}
=
A_t^{(i),\mathrm{critic}}\mathcal{V}^{\mathrm{critic}}
\in\mathbb{R}^{L\times d_{\mathrm{att}}^{\mathrm{critic}}}.
\]
Thus, for each level $\ell$, the row $C_{t,\ell}^{(i),\mathrm{critic}}$ is a convex combination of transformed market-wide queue tokens across all agents and all tracked levels, with weights determined by the compatibility between the current agent's queue state at level $\ell$ and the full market-wide queue configuration.

\paragraph{Value features and MoE value head.}
The critic feature vector is
\[
x_t^{(i)}=\big[h_{t,\mathrm{sc}}^{(i)},\ C_{t,\ell}^{(i),\mathrm{critic}}\big]\in\mathbb{R}^{d_{\text{sc\_out}}^{\text{crit}}+d_{\text{att}}^{\text{crit}}L}.
\]
In the implementation the value is produced by a $2\times 2$ mixture-of-experts (four experts), gated by
(i) the sign of the current agent’s inventory and
(ii) the sign of \textbf{global mispricing}.

Each expert is a two-layer $\tanh$ MLP mapping $x_t^{(i)}$ to a scalar:
\[
v_t^{(i,e)} = f_V^{(e)}(x_t^{(i)})\in\mathbb{R},
\qquad e\in\{1,2,3,4\}.
\]
The mispricing proxy is computed from the global scalars as
\[
\Delta_{t}^{\mathrm{critic}} =
M_t - S_t.
\]
Inventory for gating is extracted from the centralized per-agent block using $\mathbf{e}^{(i)}$:
\[
I_t^{(i)} = \sum_{j=1}^{N} e_j^{(i)}\,q_t^{(j)}.
\]
Using the same hard/soft gating map $\mathcal{G}$ as in the actor (with different slopes), obtain
\[
(w_{\mathrm{neg}}^{(i)},w_{\mathrm{pos}}^{(i)})=\mathcal{G}(I_t^{(i)}),
\qquad
(w_{\mathrm{under}},w_{\mathrm{over}})=\mathcal{G}(\Delta_{t,\mathrm{critic}}),
\]
and form the $2\times 2$ mixture weights
\[
\big(w^{(i)}_{\mathrm{neg,under}},w^{(i)}_{\mathrm{neg,over}},w^{(i)}_{\mathrm{pos,under}},w^{(i)}_{\mathrm{pos,over}}\big)
=
\big(w_{\mathrm{neg}}^{(i)}w_{\mathrm{under}},\ w_{\mathrm{neg}}^{(i)}w_{\mathrm{over}},\ w_{\mathrm{pos}}^{(i)}w_{\mathrm{under}},\ w_{\mathrm{pos}}^{(i)}w_{\mathrm{over}}\big).
\]
The final value is the corresponding convex combination:
\[
V_t^{(i)} =
w^{(i)}_{\mathrm{neg,under}}\,v_t^{(i,1)}
+
w^{(i)}_{\mathrm{neg,over}}\,v_t^{(i,2)}
+
w^{(i)}_{\mathrm{pos,under}}\,v_t^{(i,3)}
+
w^{(i)}_{\mathrm{pos,over}}\,v_t^{(i,4)}.
\]

\section{Parameters}

This section collects the numerical settings used throughout the simulations and training. In particular, the reported optimization hyperparameters include the Kullback--Leibler (KL) target schedule, the number of stochastic gradient descent (SGD) epochs per update, and the generalized advantage estimation (GAE) parameter used within the proximal policy optimization (PPO) procedure.

\begin{table}[htbp]
\centering
\begin{tabularx}{\textwidth}{@{} l l l X @{}}
\toprule
Parameter & Symbol & Value & Comment \\
\midrule
Number of bins & $K$ & $3$ &  \\
Bins of volumes & -- & $[0,40],\ [40,80],\ [80,200]$ &  \\
Number of exponentials & $R$ & $2$ &  \\
Self-excitation kernel norms (same side) & $A^{AA} = A^{BB}$ & $(0.3,\,0.12,\,0.07)$ & Scaled by $s_\alpha =0.1$ \\
Cross-excitation kernel norms (across sides) & $A^{AB} = A^{BA}$ & $(0.07,\,0.035,\,0.02)$ & Scaled by $s_\alpha =0.1$ \\
Exponential decay parameters & $\beta$ & $(3.7,\,4.2)$ &  \\
Exponential mixture weights & $w$ & $(0.7,\,0.3)$ &  \\
Baseline intensities & $\mu$ & $(3.5,\,2.5,\,1.5)$ &  \\
Shape parameter & $\eta$ & $2.8$ &  \\
Scale parameter & $\zeta$ & $70$ &  \\
\bottomrule
\end{tabularx}
\caption{Hawkes process parameters.}
\label{tab:hawkes_params}
\end{table}

\begin{table}[htbp]
\centering
\small
\setlength{\tabcolsep}{4pt}
\renewcommand{\arraystretch}{1.20}
\begin{tabularx}{\textwidth}{@{} l X c c l @{}}
\toprule
Regime mode & Description & Switching time & Changes & Regime values \\
\midrule

\texttt{Regime model 2}
& One change at a fixed time to a random regime
& $t_c=50$
& $1$
& $(\psi,\sigma)\in\Psi\times\Sigma$ \\

\texttt{Regime model 4}
& One change at a random time to a random regime
& $t_c\sim\mathcal{U}\{10,\ldots,90\}$
& $1$
& $(\psi,\sigma)\in\Psi\times\Sigma$ \\

\texttt{Regime model 5}
& Fixed number of changes on a fixed grid to random regimes
& $t_c\in\{20,50,80\}$
& $3$
& $(\psi,\sigma)\in\Psi\times\Sigma$ \\

\texttt{Regime model 6}
& Fixed number of changes on a random grid to random regimes
& $t_c^{(m)}\sim\mathcal{U}\{10,\ldots,90\}$
& $3$
& $(\psi,\sigma)\in\Psi\times\Sigma$ \\

\texttt{Regime model 7}
& Exponential waiting times for random regimes
& $\Delta t_c\sim\mathrm{Exp}(\lambda)$
& Random
& $(\psi,\sigma)\in\Psi_{\lambda}\times\Sigma$ \\
\bottomrule
\end{tabularx}
\caption{Regime-switching specifications used in the non-stationary market experiments. The variable $\psi$ denotes the informed ratio and $\sigma$ denotes the market-volatility parameter. Regime models draw the post-change state from the full regime grid $\Psi=\{0.00,0.01,0.05,0.10,0.20,0.30,0.40,0.50,0.60,0.70,0.80,0.90,0.95,0.99,1.00\}$ and $\Sigma=\{0.001,0.005,0.010,0.020\}$. For each regime model, both known and unknown variants can be evaluated; in the known variant, the current regime variables are included in the agents' observations, whereas in the unknown variant they are hidden from the learning pipeline.}
\label{tab:regime_model_specifications}
\end{table}

\begin{table}[htbp]
\centering

\begin{tabularx}{\textwidth}{@{} l c c c c @{}}
\toprule
Policy & Agents & $\phi$ & $\alpha$ & $\sigma_{\text{signal}}$ \\
\midrule
Conservative & $4$ & $3\cdot 10^{-4}$ & $3\cdot 10^{-3}$ & $0$  \\

Quasi-informed conservative & $4$ & $3\cdot 10^{-4}$ & $3\cdot 10^{-3}$ & $10^{-5}$  \\

Aggressive & $4$ & $10^{-4}$ & $10^{-3}$ & $0$  \\

Quasi-informed aggressive & $4$ & $10^{-4}$ & $10^{-3}$ & $10^{-5}$  \\
\bottomrule
\end{tabularx}
\caption{Agent configuration and policy-specific parameters. Additionally, the learning rate for both informed and uninformed conservative policies is scheduled as $\{(0,9\cdot10^{-4}),\,(15,3\cdot10^{-4}),\,(40,10^{-4}),\,(100,5\cdot10^{-5})\}$, while for both aggressive policies it is scheduled as $\{(0,8\cdot10^{-4}),\,(15,3\cdot10^{-4}),\,(40,10^{-4}),\,(100,5\cdot10^{-5})\}$. Values between the scheduled points are linearly interpolated for both schedules. In each pair, the first entry denotes the iteration and the second the scheduled value.}
\label{tab:agent_policy_params}
\end{table}

\begin{table}[htbp]
\centering

\begin{tabularx}{\textwidth}{@{} l l X @{}}
\toprule
Parameter & Value \\
\midrule
Total number of agents & $16$ \\
Agents posting volumes & $[1, 5, 10, 20]$\\
KL target schedule &
$\{(0,0.08),\,(10,0.08),\,(25,0.02),\,(45,0.01),\,(100,0.005)\}$\\
Entropy schedule &
$\{(0,0.003),\,(80,0)\}$ \\
Maximum resting quantity per agent per side & $60$ &  \\
Observation history length & $8$ &  \\
Protected levels from cancellation & $2$  \\
Postable levels & $4$ &  \\
Observable levels before cancellation & $6$ &  \\
\bottomrule
\end{tabularx}
\caption{Common agent constraints and schedules. Values between the scheduled points are linearly interpolated for all schedules. In each pair, the first entry denotes the iteration and the second the scheduled value.}
\label{tab:common_agent_params}
\end{table}

\begin{table}[htbp]
\centering
\begin{tabularx}{\textwidth}{@{} l l X @{}}
\toprule
Parameter & Symbol & Value \\
\midrule
Training batch size & -- & $81920$ \\
Minibatch size & -- & $40960$ \\
Number of SGD epochs & -- & $10$ \\
Value loss coefficient & $c_{\mathrm{vf}}$ & $1$ \\
KL penalty coefficient & $c_{\mathrm{kl}}$ & $0.2$ \\
Discount factor & $\gamma$ & $0.95$ \\
GAE parameter & $\lambda$ & $0.95$ \\
Policy clip parameter & $\varepsilon_{\mathrm{clip}}$ & $0.15$ \\
Value clip parameter & $\varepsilon_v$ & $2$ \\
\bottomrule
\end{tabularx}
\caption{PPO and training parameters.}
\label{tab:ppo_params}
\end{table}

\begin{table}[htbp]
\centering
\begin{tabularx}{\textwidth}{@{} l l l @{}}
\toprule
Parameter & Symbol & Value \\
\midrule
Interest rate & $r$ & $0$ \\
Drift & -- & $0$ \\
Rebate & $R_{\text{rebate}}$ & $0.002$ \\
Tick size & $\Delta p$ & $0.01$ \\
Initial price & $S_0 = M_0$ & $\$1$ \\
\bottomrule
\end{tabularx}
\caption{Market setup.}
\label{tab:market_params}
\end{table}

\begin{table}[htbp]
\captionsetup{list=no}
\centering
\begin{tabularx}{\textwidth}{@{} X l l @{}}
\toprule
Parameter & Symbol & Value \\
\midrule
Scalar encoder's output dimension in actor's architecture & $d_{\text{sc\_out}}$ & $64$\\
Key dimension in actor's architecture's cross-attention & $d_k$ & $64$\\
Output dimension in actor's architecture cross-attention & $d_{\text{att}}$ & $32$\\
GRU hidden size & $d_h$ & $32$\\
Policy head hidden width & $d_{\text{head}}$ & $256$\\
Scalar encoder's output dimension in critic's architecture & $d_{\text{sc\_out}}^{\text{crit}}$ & $64$\\
Key dimension in critic's architecture cross-attention & $d_k^{\text{crit}}$ & $64$\\
Output dimension in critic's architecture cross-attention & $d_{\text{att}}^{\text{crit}}$ & $32$\\
\bottomrule
\end{tabularx}
\caption{Neural network parameters.}
\label{tab:nn_params}
\end{table}
\end{appendices}
\clearpage

\bibliography{refs}

\end{document}